%% file: main.tex
\documentclass[10pt, conference]{IEEEtran}

\usepackage{cite}
\usepackage{amsmath,amssymb,amsfonts,amsthm}
\usepackage{algorithmic}
\usepackage{graphicx}
\usepackage{textcomp}
\usepackage{xcolor}
\usepackage{xspace}
\def\BibTeX{{\rm B\kern-.05em{\sc i\kern-.025em b}\kern-.08em
    T\kern-.1667em\lower.7ex\hbox{E}\kern-.125emX}}

\newtheorem{theorem}{Theorem}[section]
\newtheorem{definition}[theorem]{Definition}
\newtheorem{example}[theorem]{Example}
\newtheorem{corollary}[theorem]{Corollary}
\newtheorem{proposition}[theorem]{Proposition}
\newtheorem{lemma}[theorem]{Lemma}
\newcommand{\gnn}{\textsc{Gnn}\xspace}
\newcommand{\gnns}{\textsc{Gnns}\xspace}

\begin{document}

\title{A Game for Counting Logic Formula Size and an Application to Linear Orders}

\author{
    \IEEEauthorblockN{Gr\'egoire Fournier}
    \IEEEauthorblockA{
        \textit{University of Illinois at Chicago} \\
        Chicago, USA \\
        gfourn2@uic.edu
    }
    \and
    \IEEEauthorblockN{Gy\"orgy Tur\'an}
    \IEEEauthorblockA{
        \textit{University of Illinois at Chicago} \\
        Chicago, USA \\
        \textit{ELRN Research Group on AI} \\
        Szeged, Hungary \\
        gyt@uic.edu
    }
}
\maketitle

\begin{abstract}
Ehrenfeucht - Fraïssé (EF) games are a basic tool in finite model theory for proving definability lower bounds, with many applications in complexity theory and related areas.
They have been applied to study various logics, giving insights on quantifier rank and other logical complexity measures. In this paper, we present an EF game to capture formula size in counting logic with a bounded number of variables.
The game combines games introduced previously for counting logic quantifier rank due to Immerman and Lander, and for first-order formula size due to Adler and Immerman, and Hella and V\"a\"an\"anen. The game is used to prove the main result of the paper, an extension of a formula size lower bound of Grohe and Schweikardt for distinguishing linear orders, from 3-variable first-order logic to 3-variable counting logic. As far as we know, this is the first formula size lower bound for counting logic.
\end{abstract}


\input{paper/1_introduction}
\input{paper/2_related_work}
\input{paper/3_background}
\input{paper/4_game}
\input{paper/5_characterization}
\input{paper/6_theorems}
\input{paper/7_full_proof}
\input{paper/8_conclusion}

\bibliographystyle{acm}
\bibliography{references}

\clearpage

\input{paper/Appendix}

\end{document}

%% file: paper/1_introduction.tex
\section{Introduction}
\label{sec:Introduction}
Ehrenfeucht - Fraïssé (EF) games \cite{Fraisse54, Ehrenfeucht1961} are a basic tool of finite model theory for proving definability lower bounds~\cite{book_libkin,ebbing_flum}.
Combined with logical characterizations of complexity classes, they provide a logic-based approach to problems in complexity theory.
The original form of EF games gives bounds for quantifier rank in first-order logic 
but the games have been extended and modified for many logics and formula complexity measures.
An EF game for \emph{formula size} in first-order logic (FO) is given by Hella and V\"a\"an\"anen~\cite{Hella2012}, building on Adler and Immerman~\cite{n!_lwb_fmulasize}. 
In what follows, we refer to this game as the HV-game.

Counting logic extends first-order logic by adding the counting quantifier $\exists^{\geq k}$, and is frequently used in complexity theory and combinatorics. Counting logic turns out to be relevant for understanding the computational power of graph neural networks (\gnn) as well \cite{logic_gnn}. An EF game for distinguishing graphs in counting logic with a bounded number of variables is formulated by Immerman and Lander~\cite{immerman1990describing}. The game~\cite{immerman1990describing} extends the basic EF setup by an additional phase in each round involving the choice of subsets of the same cardinality in the two structures.

Understanding \emph{formula size in counting logic with a bounded number of variables} would be useful, in particular, for a further analysis of the logical characterizations of \gnns~(Barcel\'o et al.~\cite{Barcelo20}). An EF game for this setup could be a useful tool in this endeavor.

While the standard EF game is played on two structures, the HV game is played on two \emph{sets of structures}, referred to in this paper as \emph{families}. 
Grohe and Schweikardt~\cite{grohe_succinc}
prove a formula size lower bound for linear orders for the 3-variable fragment of first-order logic, using this technique implicitly.
They show that every first-order 3-variable formula that distinguishes a linear order of size $n$ from a larger one has size $\Omega(\sqrt{n})$. Their proof
is based on the notion of a \emph{separator} and an involved weighting technique, which allows for a refined analysis of the syntax tree of a formula.

The separator and the weighting scheme are useful tools to gauge the progress made in the subformulas of a distinguishing formula. The proof is a detailed case analysis, with numerous subcases required to deal with quantifiers.

\bigskip

In this paper, we formulate a game for capturing formula size complexity for counting logic (Theorem~\ref{theorem:counting_size}). The game is a combination of the Immerman-Lander and Hella-V\"a\"an\"anen games mentioned above. Restricted versions characterize formula size for fragments of counting logic where the \emph{number of variables} and the \emph{counting rank} are bounded.

The main result of the paper is
that every 3-variable counting logic formula with counting rank $t$ distinguishing a linear order of size $n$ from a larger one has size at least $\sqrt{n}/t$ (Theorem~\ref{theorem:result_counting}). 
This result extends the result of~\cite{grohe_succinc}
from first-order logic to counting logic. The theorem is proved using the game characterization. There is a simple distinguishing formula of size $n/t$ (Proposition ~\ref{proposition:upper_bound}). In addition, for the case $t=1$, i.e., for 3-variable FO, our result improves the formula size lower bound of~\cite{grohe_succinc} from $\sqrt{n}/2$ to $\sqrt{n}$. 


The lower bound proof uses modifications of the separators and the weighting scheme. Some cases considered are identical to~\cite{grohe_succinc}. The overall argument is, however, different. The main difference is in the most technical part, the proof of the quantifier case. We introduce \textit{gap sets} and \textit{gap variables} (Definitions~\ref{definition:gap} and \ref{definition:gap_variables}), which makes the argument more similar to standard EF arguments and generalizes the reasoning. The new proof of the FO case is somewhat simpler than the original.
A tree summarizing the proof structure and the different cases is given in Fig. \ref{fig:proof_structure}.


The paper is structured as follows. 
After reviewing related work in Section~\ref{sec:rel}, we describe the counting and HV games in Section~\ref{sec:back}. Section~\ref{sec:prel} describes the game and Section~\ref{sec:gchar} gives the correspondence between the game and counting logic formula size. Section~\ref{sec:linear_order_lower_bound} contains the application on linear orders, with the proof of the lower bound in Section~\ref{section:proof}. Our main technical contribution, the proof of the main Lemma on counting quantifiers, is presented separately in Section \ref{sec:sep}. Finally, Section~\ref{sec:conc} contains remarks and directions for future research.

%% file: paper/2_related_work.tex
\section{Related work} \label{sec:rel}

Counting logic has been discussed in several different forms.
Grohe~\cite{logic_gnn} defines counting logic ${\mathcal C}$ as first-order logic (FO) extended by counting quantifiers of the form $\exists^{\geq k} x$, and ${\mathcal C}_m$ as its fragment using at most $m$ variables. This is the counting logic we consider in this paper. Previous work using this kind of counting logic includes Immerman and Lander~\cite{immerman1990describing} and Cai et al.~\cite{opt_graph_id}. Grohe~\cite{Grohe23}, on the other hand, considers a more powerful counting logic, where formulae can include arithmetic operations on the number of elements satisfying a formula (see also Kuske and Schweikardt~\cite{Kuske17}).

EF games using sets of structures, capturing the number of quantifiers as opposed to quantifier rank, have been proposed by Immerman~\cite{Immerman81}. These \textit{multi-structural (MS)} games receive increasing current attention (Fagin et al.\cite{fagin2022number}, Carmosino et al.~\cite{carmosino2023finer}, Vinall-Smeeth ~\cite{vinallsmeeth2024quantifierdepthquantifiernumber}). HV games are essentially extensions of MS games, also modelling Boolean connectives in the formulae.

A graph neural network (\gnn) is a variant of  neural networks for machine learning problems involving graphs~\cite{gnn_first}. Such a network allows the use of deep learning techniques to classify graphs (graph classification), or to classify the nodes of a large graph (node classification). The computational power of \gnn is closely related to the Weisfeiler - Leman (WL) graph isomorphism algorithm (Morris et al.~\cite{Morris19}, Xu et al.~\cite{Xu19}). The connection of the WL algorithm to counting logic with a bounded number of variables~(\cite{opt_graph_id,logic_gnn}) brings these logics into the \gnn picture as well. Barcel\'o et al.~\cite{Barcelo20} gave logical characterizations in terms of counting logic with a bounded number of variables using results established in modal logic (Otto~\cite{otto_counting_bisimulation}). The complexity aspects of the characterizations are not discussed in~\cite{Barcelo20}, and studying this aspect (also pointed out in Grohe~\cite{logic_gnn}) has been a motivation for the topic of this paper (a brief further discussion is given at the end of the paper).

The \gnn characterizations of Grohe~\cite{Grohe23}
establish a connection of \gnn to threshold circuits,
a Boolean circuit model of neural networks. The computational power of such circuits corresponds to counting logic with an arbitrary built-in predicate. Proving superpolynomial lower bounds for threshold circuits is an open problem. Hajnal et al.~\cite{Hajnal87,Hajnal93} prove an exponential lower bound for depth-2 circuits with polynomial weights,
and so far this lower bound has not been extended to either depth-2 with unrestricted weights or to higher depths. The same papers prove a quantifier rank lower bound for counting logic with successor as the built-in relation. Similar results are also given in Etessami~\cite{Etessami95,Etessami97}.  Karchmer and Wigderson~\cite{Karchmer90} formulate an approach, related to HV games, to proving monotone formula depth lower bounds. They also prove a depth version of the Krapchenko formula size lower bound.
Krapchenko's Theorem is proved in ~\cite{Hella2012}
as an application of HV games.

General background for the topic of this paper is given in Immerman~\cite{Immerman_desc}, Ebbinghaus and Flum~\cite{ebbing_flum}, Libkin~\cite{book_libkin}, Otto~\cite{otto_book} and Hamilton~\cite{Hamilton20}.

%% file: paper/3_background.tex
\section{Background} \label{sec:back}

In this section we introduce basic notation used in the paper and review EF and HV games.

\subsection{Basic definitions}

\subsubsection{Logics}

We consider relational structures over a fixed vocabulary.\\
Counting logic $\mathcal{C}$ is obtained by extending first-order logic with counting quantifiers $\exists^{\ge k} x \, \phi(x)$ and $\forall^{\geq k} x \, \phi(x)$. Here $\exists^{\ge k} x \phi(x)$ 
means that there are at least $k$ distinct assignments to the variable $x$ that satisfy $\phi$. Thus $\exists^{\ge k} \, \phi(x)$ is logically equivalent to $\exists x_1 \dotsc \exists x_k (\bigwedge_i \phi(x_i) \wedge \bigwedge_{i,j} x_i\neq x_j)$.
The quantifier $\forall^{\geq k} x \phi(x)$ stands for $\neg \exists^{\geq k} x \neg \phi(x)$.
In $\exists^{\geq k} x \phi(x)$ and $\forall^{\geq k} x \phi(x)$, $k$ is referred to as the \emph{counting rank} of the quantifier. The counting rank of a formula is the maximum counting rank of its quantifiers. 

As a counting quantifier can be replaced by standard quantifiers, adding counting quantifiers does not change the expressivity of first-order logic.
It does, however, impact the succinctness, the minimum size of formulae expressing a property.
In applications to finite model theory one usually considers a sequence $({\mathcal A}_n, {\mathcal B}_n)$ of pairs of structures. 
Complexity bounds to be proven are also functions of $n$. Note that $\forall^{\geq k} x \, \varphi(x)$
is equivalent to $\exists^{\geq n - k + 1} x \, \varphi(x)$ for an $n$-element structure. 
The transformation increases counting rank and thus it cannot be used in formula size bounds for the bounded counting rank case. 

Parameters to be considered are the bound $m$ on the number of variables,
the bound $t$ on the counting rank, and the bound $w$ on the formula size. The fragment of counting logic of formulae containing at most $m$ variables and counting rank at most $t$ is denoted by $\mathcal{C}^t_m$.
\subsubsection{Structures and Families}

The universe of a structure $\mathcal{A}$ is denoted by $\mathcal{U}^\mathcal{A}$. We use $x_j$, $j \in \mathbb{N}$, to denote variables. A variable assignment for a structure $\mathcal{A}$ is a finite partial mapping $\alpha : \mathbb{N} \to  \mathcal{U}^\mathcal{A}$. The finite
domain of $\alpha$ is denoted by $dom(\alpha)$.

An \emph{interpretation} is a pair $(\mathcal{A}, \alpha)$. For a formula $\phi$, $(\mathcal{A}, \alpha) \models \phi$ means that the assignment $\alpha$ satisfies the formula $\phi$ in the structure $\mathcal{A}$, with  $dom(\alpha)$ containing all the $j$ for which the variable $x_j$ is free in $\phi$. A formula $\phi$ \emph{distinguishes}  interpretations $(\mathcal{A}, \alpha)$ and $(\mathcal{B}, \beta)$, denoted by $((\mathcal{A}, \alpha), (\mathcal{B}, \beta)) \models \phi$,  if $(\mathcal{A}, \alpha) \models \phi$ and $(\mathcal{B}, \beta) \models \neg\phi$.\\
A \emph{family} $A_{\mathcal{A}, D}$ is a set of interpretations $\{(\mathcal{A}, \alpha_i) | \, i \in \Gamma \}$ where
    $\mathcal{A}$ and $D = dom(\alpha_i)$ are fixed and $\Gamma$ is some set. When the context is clear, we drop the subscript.
    We write $(A, B) \models \phi$ to express that for all $(\mathcal{A}, \alpha) \in A$, $(\mathcal{A}, \alpha)  \models \phi$ and  for all $ (\mathcal{B}, \beta) \in B$, $(\mathcal{B}, \beta)  \models \neg \phi$, and we say that $\phi$ distinguishes $(A,B)$. For a structure $\mathcal{A}$, we  denote by $A_0$ the family $\{(\mathcal{A}, \emptyset)\}$,  containing a single interpretation with 
    the empty assignment.

\subsubsection{Operations}
If $\alpha$ is an assignment on $\mathcal{A}$, $a \in \mathcal{U}^{\mathcal{A}}$ and $j \in \mathbb{N}$, then $\alpha(a/j)$ is
the assignment that maps $j$ to $a$ and agrees with $\alpha$ otherwise.

Given a family $A$, a \emph{choice function} is of the form 
$F \, : \, A \to \mathcal{U}^{\mathcal{A}}$.
The set of all choice functions on the family $A$ is denoted by $F_A$. We define two operations on families~\cite{Hella2012}. 

\begin{itemize}
\item \emph{Change}: Given a family $A$, a choice function $F \in F_A$ and $j \in \mathbb{N}$, the change operation on $A$ with $F$ for the variable $x_j$ produces the family \[A(F/j):= \{(\mathcal{A}, \alpha(F(\mathcal{A},\alpha)/j)) : (\mathcal{A}, \alpha) \in A\}.\] 
In the new family, the assignment to $x_j$ is changed based on the choice function $F$.
If $j \not\in D$ then $x_j$ is a new variable, and $D$ is updated to $D \cup \{j\}$. Thus a change operation may either leave the domain unchanged or add a new element to it.

\item \emph{Multiplication}: Given a family $A$ and $j \in \mathbb{N}$, the multiplication operation $A$ for the variable $x_j$ produces the family
\[A(*/j):= \{(\mathcal{A}, \alpha(a/j)) : (\mathcal{A}, \alpha) \in A, a \in \mathcal{U}^\mathcal{A}\}.\]
The new family consists of interpretations with $x_j$ assigned to all possible values in
$\mathcal{U}^\mathcal{A}$. Here, again, the domain is either unchanged or a new element is added to it.
\end{itemize}

\subsubsection{Formula complexity}The size of a formula is defined inductively: if $\phi$ is an atomic formula, $|\phi|=1$; and for FO formulae $\phi,\psi$, $|\neg \phi|=|\phi| + 1$; $|\phi \vee \psi|= |\phi \wedge \psi|=|\phi|+|\psi|$; $|\exists x_j \phi|= |\forall x_j \phi| = 1+ |\phi|$. For the definition of the quantifier rank of a formula, we refer the reader to \cite{ebbing_flum}.

\subsection{Review of games}
\label{games_review}
In this section we review the counting logic 
game~\cite{immerman1990describing} and the first-order logic formula size game~\cite{Hella2012}, referred to as the HV game.

\begin{definition}[$r$-round EF $m$-pebbling game]
\label{pebble_game}
The game $EF(\mathcal{A},\mathcal{B})$ is played on two relational structures $\mathcal{A}$ and $\mathcal{B}$.  There are two players, Spoiler and Duplicator, and $m$ pairs of pebbles $(a_i,b_i)$ for $i \in [m]$. It goes as follows:
\begin{itemize}
    \item For $r$ rounds:
    \begin{itemize}
    \item Spoiler picks a set of elements $S_1$ of $\mathcal{A}$ or $\mathcal{B}$ and a number $i \in [m]$.
    Duplicator selects a set $S_2$ of elements of the same cardinality in the other structure.
    \item Spoiler picks an element in $S_2$ and Duplicator selects an element in $S_1$. The pebble $a_i$ (resp., $b_i$) holds the value of the element picked in $\mathcal{A}$ (resp., $\mathcal{B}$). 

    \end{itemize}
    \item The $r$-round game ends in the position $\vec{a}=(a_1,\dotsc, a_m)$, $\vec{b}=(b_1,\dotsc, b_m)$.
    Duplicator wins if the mapping from\\ $((\vec{a},\vec{c}^\mathcal{A})$ to $(\vec{b},\vec{c}^\mathcal{B}))$ is a partial isomorphism between $\mathcal{A}$ and $ \mathcal{B}$, where $\vec{c}$ denotes the constants of the language. 

\end{itemize}
\end{definition}
\begin{theorem}
\label{theorem:EF_counting} 
The following are equivalent:
\begin{itemize}
\item $\mathcal{A}$ and $\mathcal{B}$ satisfy the same sentences of $\mathcal{C}_m$ of quantifier rank at most $r$.
\item Duplicator has a winning strategy in the $r$-round $m$-pebbling game.
\end{itemize}
\end{theorem}

A similar result holds for the bounded counting rank fragment ${\mathcal C}_m^t$, by restricting the cardinality of the sets picked to be at most $t$.

\begin{figure}[h]
  \centering
  \includegraphics[width=\linewidth]{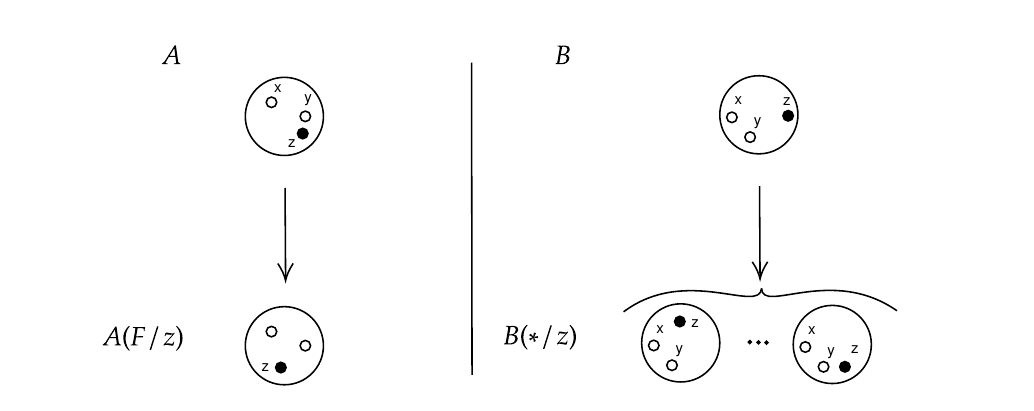}

\caption{Two families, $A = \{(\mathcal{A}, \alpha)\}$ and $B = \{(\mathcal{B}, \beta)\}$, for the $\exists z $ move in the HV-game.}
\label{figure:hella_vaananen_game_picture}
\end{figure}

\begin{definition}[HV game for formula size on first-order logic]
The game HV$_w(A, B)$ is played on two families, $A$ and $B$, by Spoiler and Duplicator. The initial position is $(w, A, B)$. There are 
five possibilities for the continuation of the game:
\begin{itemize}
    \item $\neg$-move: the game continues from the position $(w-1,B,A)$.
    \item $\bigvee$-move: Spoiler chooses $1 \leq u, v < w$ such that $u + v = w$ and partitions $A$ to get $ C \cup D$. Duplicator chooses the next position as $(u, C, B)$ or $(v, D, B)$ .
    \item $\bigwedge$-move: similar but played on $B$.
    \item $\exists$-move: Spoiler chooses $j \in \mathbb{N}$ and a choice function $F$ from $F_A$. Then the game continues from $(w-1, A(F/j), B(*/j))$.
    \item $\forall$-move: similar but Spoiler chooses on $B$.
        \end{itemize}
Spoiler wins if the game reaches a position $(w, A, B)$ for $w \geq 1$ and there is an atomic formula 
that distinguishes $A$ and $B$. Duplicator wins if the game reaches a position $(1, A, B)$ and Spoiler does not win.
\end{definition}

\begin{theorem}
\label{theorem:FO_size}
Let $(A, B)$ be a pair of families, and let $w$ be a positive integer. Then the following  are equivalent:
\begin{enumerate}
    \item Spoiler has a winning strategy in the game HV$_w(A, B)$.
    \item There is a formula $\phi$ of FO of size $|\phi| \le w$
    such that $(A, B) \models \phi$.
\end{enumerate}
\end{theorem}

Note that for any $m \in \mathbb{N}$, one can define the variant of the game HV$_w^m(A, B)$, for which the moves $\exists$ and $\forall$ are restricted by $j \in \{1, \dotsc, m\}$. This game then characterizes $FO_m$, the fragment of FO logic,, with $m$ variables. An illustration of the $\exists$-move is given in Fig. \ref{figure:hella_vaananen_game_picture}.

In Table \ref{tab:logic_notation}, we summarize the game characterizations discussed above and the game we are about to introduce.

\begin{table}[htbp]
    \caption{Complexity characterizing games.}
\begin{center}
    \renewcommand{\arraystretch}{1.5} 
    \begin{tabular}{|c|c|c|}
        \hline
         & FO$_m$ & \textbf{$\mathcal{C}_m$} \\
        \hline
        \textbf{Quantifier Rank} & \cite{Fraisse54}, \cite{Ehrenfeucht1961} & EF$_m$ \cite{immerman1990describing}\\
        \hline
        \textbf{Size} & HV$_m$ \cite{HELLA1996} & CS$_m$ (new) \\
        \hline
    \end{tabular}
    \end{center} 
    \label{tab:logic_notation}
\end{table}

%% file: paper/4_game.tex
\section{The counting logic formula size game} \label{sec:prel}

In this section we define our game for counting logic formula size. We start by extending the operations of the HV game to this setting.

\subsection{Extended operations}
\label{sec:extended_operations}

\begin{definition} \label{def:kc} ($k$-choice function, selection)
Given a family $A$, a \emph{$k$-choice function} is of the form
$F^k=(F^k_1,F^k_2,\dotsc,F^k_k)$,
where each $F^k_i$ is a choice function, and for every $(\mathcal{A}, \alpha) \in A$ the elements $F^k_i(\mathcal{A}, \alpha)$ are pairwise distinct.
The set of all the $k$-choice functions on the family $A$ is denoted by $F_A^k$. A choice function $F$ is a \emph{selection} from $F^k$ if $F(\mathcal{A},\alpha) \in \{ F^k_1(\mathcal{A},\alpha),\dotsc,F^k_k(\mathcal{A},\alpha)\}$ for every 
$(\mathcal{A},\alpha) \in A$. We say that $F$ is \emph{selected} from $F^k$.
\end{definition}

Note that if $F$ is a \emph{selection} from $F^k$, then the choices on different interpretations may correspond to different $F^k_i$s. This concept will be used to represent the Spoiler's choices of one interpretation from a set of interpretations. The extended set of operations is the following:

\begin{itemize}
\item \emph{$k$-Change}: Given a family $A$ and $j, k \in \mathbb{N}$, the $k$-change operation associated to $F^k \in F^k_A$ produces the family
\[A(F^k/j):= \{(\mathcal{A}, \alpha(F^k_i(\mathcal{A},\alpha)/j)) : (\mathcal{A}, \alpha) \in A, 1 \le i \le k\}.\]
Thus each interpretation gives rise to $k$ interpretations composing the new family, where the assignments to $x_j$ are changed based on the $k$-choice function $F^k$.

\item \emph{$k$-Multiplication}: Given a family $A$ and $j \in \mathbb{N}$, the $k$-Multiplication operation on $A$ does the following: for every $k$-choice function $F^k$ on $A$, a choice function $F$ is \emph{selected} from $F^k$. The union of the $A(F/j)$ over $F^k \in F^k_A$ forms $A(*^k/j)$. 

In other words, for every $F^k \in F^k_A$ and $ (\mathcal{A}, \alpha) \in A$, one interpretation is picked from $\{(\mathcal{A}, \alpha(F^k_i(\mathcal{A},\alpha)/j)): 1 \le i \le k\}$ to be part of $A(*^k/j)$. Thus each interpretation generates $card(F^k_A)$ interpretations that compose the family $A(*^k/j)$.
\end{itemize}

\subsection{Formula size game for counting logic}
\label{sec:complete}

We are now ready to define the game for counting logic formula size (referred to as the CS game for ``counting size''). It is presented in the ${\mathcal C}_m$ version as this will be used in the rest of the paper. The new $\exists^{\geq k}$-move is illustrated in Fig. \ref{figure:size_counting_game}.

\begin{definition}[CS game for formula size on counting logic]
\label{complete_game}
The game CS$_w^m(A, B)$ has two players, Spoiler and Duplicator, $m$ is the number of variables. $A,B$ are two families with $dom(A) = dom(B)$ of size at most $m$. 
Suppose after $p$ moves we reach the position $(w, A, B)$. 
Depending on Spoiler's choice, the game continues as follows:

\begin{itemize}

    \item $\neg$-move: the game continues from the position $(w-1,B,A)$.
    
    \item $\bigvee$-move: Spoiler first chooses numbers $u$ and $v$ such that 1 $\leq$  $u, v < w$ and $u + v = w$. Then Spoiler partitions $A$ into a pair of families $C$ and $D$. The game continues either from the position $(u$, $C$, $B)$ or from the position $(v$, $D$, $B)$ according to Duplicator's choice.
    
    \item $\bigwedge$-move: Spoiler first chooses numbers $u$ and $v$ such that  $1 \leq$ $u, v < w$ and $u + v = w$. Then Spoiler turns $B$ into a pair of families $C$ and $D$. The game continues either from the position $(u$, $A$, $C)$ or from the position $(v$, $A$, $D)$ according to Duplicator's choice.

    \item $\exists^{\geq k}$-move: Spoiler chooses $j \in [m],k \in \mathbb{N}$ and a $k$-choice function $F^k$ on $A$. For every $k$-choice function $G^k$ on $B$, Spoiler selects $G$ from $G^k$. The union of the $B(G/j)$ over $G^k \in F^k_B$ forms $B(*^k/j)$. Then the game continues from the position $(w-1$, $A(F^k/j)$, $B(*^k/j))$.

    \item $\forall^{\geq k}$-move: Spoiler chooses $j \in [m],k \in \mathbb{N}$ and a $k$-choice function $F^k$ on $B$. For every $k$-choice function $G^k$ on $A$, Spoiler selects $G$ from $G^k$. The union of the $A(G/j)$ over $A^k \in F^k_A$ forms $A(*^k/j)$. Then the game continues from the position $(w-1$, $A(*^k/j)$, $B(F^k/j))$.

\end{itemize}

(Atomic) The game ends in a position $(w,A,B)$ if either there is an atomic formula $\phi$ such that $(A,B) \models \phi$, in which case Spoiler wins, or if $w = 1$, in which case Duplicator wins if there is no such $\phi$.
\end{definition}

\begin{figure}[h]
  \centering
  \includegraphics[width=\linewidth]{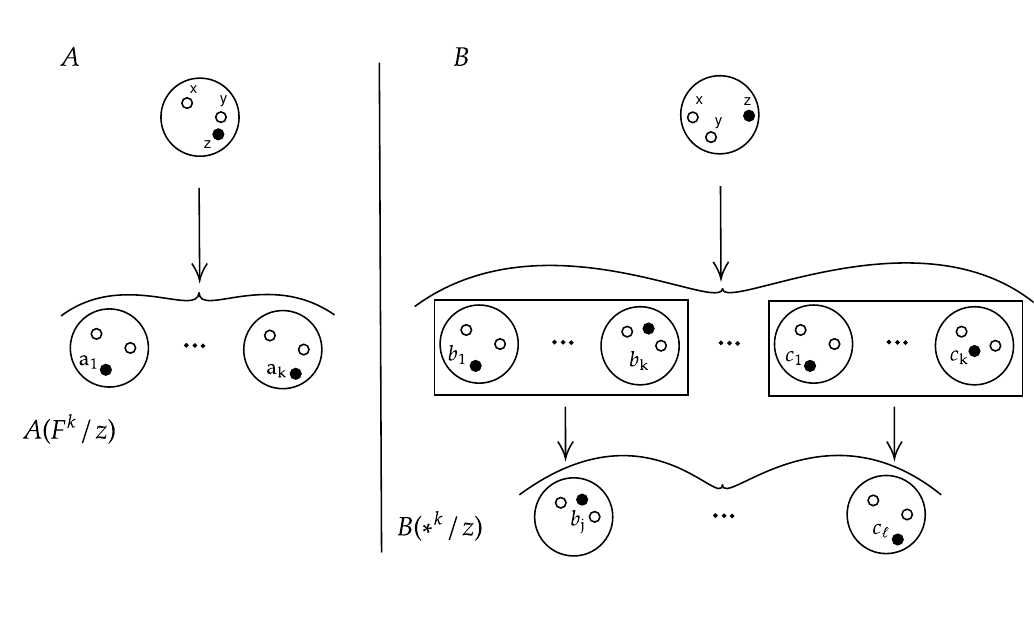}
\caption{ The $\exists^{\geq k}$-move in the CS game. On the left the $k$-Change operation: $k$ different elements are chosen by Spoiler. On the right the two steps of the $k$-Multiplication. The first step is to form all $k$-choice functions (each symbolized by a box). In the second step Spoiler picks one interpretation in each box to compose the new family.}
\label{figure:size_counting_game}

\end{figure}

\subsection{An illustration of the CS game}
\label{appendix:example}

\begin{figure}[h]
  \centering
  \includegraphics[width=\linewidth]{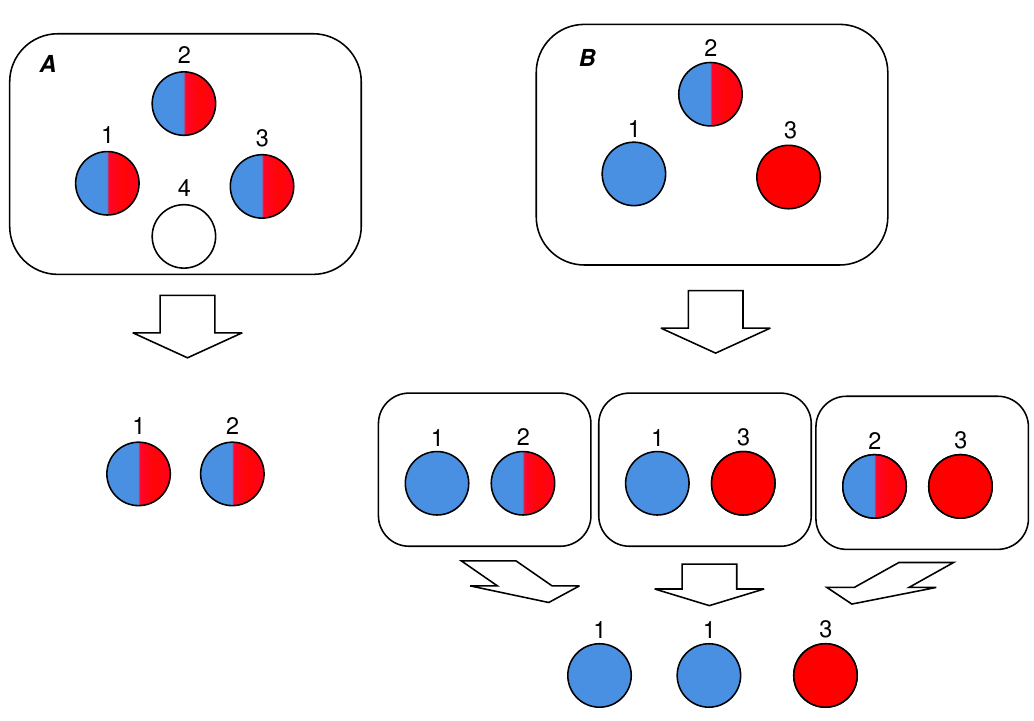}
\caption{Structures $\mathcal{A}$ and $\mathcal{B}$ with relations $red$ and $blue$, and a possible $\exists^{\geq 2}x$ move.}
\label{figure_3}
\end{figure}

\begin{example}
\label{example_1}
Consider two structures $\mathcal{A}$ and $\mathcal{B}$ of 3 and 4 elements respectively, onto which we define the unary relations $blue$ and $red$, as illustrated in Fig. \ref{figure_3}.\\
The sentence $\exists^{\ge 2}x \; ( blue(x) \wedge red(x))$ distinguishes $\mathcal{A}$ and $\mathcal{B}$ and is of size $w=3$. It is easy to check by enumerating a finite number of formulae that no formula of size $w=2$ in counting logic with counting rank $2$ is able to distinguish $\mathcal{A}$ from $\mathcal{B}$.\\
Consider the starting position  CS$^{1,2}_3(A_0,B_0)$. In the following description we omit $w$ for simplicity.
Assume again that Spoiler plays $\exists^{\ge 2} x$ and picks  $\{1,2\}$ 
in the left, corresponding to two copies of $\mathcal{A}$, with $\alpha(x) = 1$ and $\alpha(x) = 2$. The $2$-choice functions on $B$ are $\{1,2\}, \{1,3\},\{2,3\}$. \\
Spoiler now has several options to select a choice function: $\{\{1,1,3\}\} = \{1,3\}$ (as on the Figure), $\{1,2\}$, $\{2,3\}$ or $\{1,2,3\}$. In all these cases, neither $red$ nor $blue$ distinguishes the two families. So the game has to go on for Spoiler to be able to win. In fact, it is not too complicated to check that Spoiler does not win from the position CS$^{1,2}_2(A_0,B_0)$.
Continuing from ($\{1,2\}$, $\{1,3\}$) as in the Figure, Spoiler can play the $\wedge$ move to get to one of the two positions:
\begin{itemize}
    \item $\{1,2\}$ on $A$ and  $\{1\}$ on $B$, which wins by playing $red$.
    \item $\{1,2\}$ on $A$ and  $\{3\}$ on $B$, which wins by playing $blue$.
\end{itemize}
So, the "$\exists^{\geq 2} x (red(x) \wedge blue(x))$" sequence of moves is winning for Spoiler and the corresponding formula separates $\mathcal{A}$ from $\mathcal{B}$.
\end{example}

%% file: paper/5_characterization.tex
\section{Game characterization theorem} \label{sec:gchar}

In this section, we state the game characterization theorem and present its corollaries before proceeding to the proof of the theorem. 

\subsection{Game characterization and corollaries}

\begin{theorem}[Characterization Theorem]
\label{theorem:counting_size}
Let $A, B$ be families and $w$ be a positive integer. Then the following are equivalent:
\begin{enumerate}
    \item Spoiler has a winning strategy in the game CS$_w^{\,m}(A, B)$.
    \item There is a ${\mathcal C}_m$ formula $\phi$ of size $|\phi| \leq w$ such that $(A, B) \models \phi$.
\end{enumerate}
\end{theorem}

A corollary for structures is as follows: 

\begin{corollary}
Let $(\mathcal{A}$ and $\mathcal{B})$  be structures and let $w$ be a positive integer. Then the following conditions are equivalent:
\begin{enumerate}
    \item Spoiler has a winning strategy in the game CS$_w^{\,m}(A_0, B_0)$.
    \item There is a ${\mathcal C}_m$-sentence $\phi$ of size $|\phi| \leq w$  such that $(\mathcal{A}, \mathcal{B}) \models \phi$.
\end{enumerate}
\end{corollary}

\begin{corollary}
Let $(\mathcal{A}$ and $\mathcal{B})$  be structures and let $w$ be a positive integer. Then the following conditions are equivalent:
\begin{enumerate}
    \item Spoiler has a winning strategy in the game CS$_w^{\,m}(A_0, B_0)$.
    \item There is a ${\mathcal C}_m$-sentence $\phi$ of size $|\phi| \leq w$  such that $(\mathcal{A}, \mathcal{B}) \models \phi$.
\end{enumerate}
\end{corollary}

We now describe the variation of the game for bounded counting rank.
\begin{definition}[CS game for ${\mathcal C}^t_m$-formula size]
\label{complete_game_bounded}
The CS$_w^{m,t}(A, B)$ game is the version of the CS$_w^{m}(A, B)$ game where $k \le t$ for the $\exists^{\geq k}$ and $\forall^{\geq k}$ moves.
\end{definition}

\begin{corollary}
\label{theorem:counting_size_bounded_counting_rank}
Suppose $A$ and $B$ be families and let $w$ be a positive integer. Then the following are equivalent:
\begin{enumerate}
    \item Spoiler has a winning strategy in the game CS$_w^{\,m,t}(A, B)$.
    \item There is a ${\mathcal C}_m^t$-formula $\phi$ of size $|\phi| \leq w$ such that $(A, B) \models \phi$.
\end{enumerate}
\end{corollary}

The proof is similar to the CS game characterisation of $\mathcal{C}_m$ formula size presented in Section \ref{appendix:proof_game}, but with the aforementioned parameter $k$ bounded, and is omitted.

We also derive a variation of the game for graded modal logic, where the modal binary relation is denoted by $E$.

\begin{definition}[CS game for graded modal logic formula size]
The $E$-CS$_w^{\,2}(A, B)$ game is the two variables version of the CS$_w^{\,2}(A, B)$ game where additionally for the $\exists^{\geq k}$ and $\forall^{\geq k}$ moves: Spoiler chooses $i \in [m]$ that has already been chosen (i.e a variable $x_j$ with an assignment), and the spaces of the $k$-choice functions are limited to $\{e \in A | E(\alpha(x_i),e)\}$ and $\{e \in B | E(\beta(x_i),e)\}$.
\end{definition}

Following the equivalence between the expressivity of \gnns and the graded modal logic established in \cite{Barcelo20}, we obtain the following corollary linking the distinguishing power of \gnns and EF games:

\begin{corollary}
Suppose $A$ and $B$ be families and let $w$ be a positive integer. Then the following are equivalent:
\begin{enumerate}
    \item Spoiler has a winning strategy in the game $E$-CS$_w^{\,2}(A, B)$.
    \item There is a \gnn that expresses an FO formula of size $w$ that is capable of distinguishing $A$ from $B$.
\end{enumerate}
\end{corollary}
\noindent
We now prove Theorem \ref{theorem:counting_size}
\subsection{Proof of the counting game for size characterisation}

\label{appendix:proof_game}

We use induction on $w$. If $w = 1$, Spoiler wins in CS$^m_1(A, B)$ only if there is an atomic formula $\phi$, so verifying $|\phi|=1$, such that $(A, B) \models \phi$. Reciprocally, if $|\phi|= 1$ and $(A, B) \models \phi$, then $\phi$ is an atomic formula and spoiler wins in CS$^m_1(A, B)$. \\
We now assume the equivalence for $v < w$.\\
To prove the forward direction of the equivalence, suppose Spoiler has a winning strategy in the game CS$^m_w(A, B)$ starting with:

\begin{itemize}
    \item $\neg$-move. Spoiler has a winning strategy in CS$^m_{w-1}(B, A)$, so by induction hypothesis there is a formula $\psi$ verifying $|\psi| \leq w - 1 $ and  $(B, A) \models \psi$. Finally $(A, B) \models \neg \psi$ and $|\neg \psi| = 1 + |\psi| \leq w$ so we get $(2)$.

    \item $\bigvee$-move, choosing $u, v, C$ and $D$ such that  $1 \leq$  $u, v < w$ with $u + v = w$, splitting $A$ into $C$ and $D$. Spoiler has a winning strategy in CS$_u^m(C, B)$ and CS$_u^m(D, B)$, so by the induction hypothesis, there are formulae $\psi$ and $\theta$ such that $|\psi| \leq u, $ $|\theta| \leq v$, $(C, B) \models \psi$ and $(D, B) \models \theta$.\\
    We have that $C \models \psi$ and $D \models \theta$, so $A \models \psi \vee \theta$. Conversely, $B \models \neg \psi$ and $B \models \neg \theta$, and consequently $B \models \neg (\psi \vee \theta)$.\\
    It follows that $(A, B) \models \psi \vee \theta$, and since $|\psi \vee \theta| = |\psi| + |\theta| \leq u + v = w$, we get $(2)$.

    \item  $\bigwedge$-move; and by a symmetrical argument to the $\bigvee$-move we get $(2)$.

    \item $\exists^{\geq k}$-move; choosing $j \in [m], k \in \mathbb{N}$, a $k$-choice function $F^k$ on $A$ and for every $G^k \in F^k_B$, selecting $G$ from $G^k$ to form $B(*^k/u) = \bigcup_{G^k \in F^k_B} B(G/u) $. Spoiler has a winning strategy in CS$^m_{w-1}(A(F^k/j)$, $B(*^k/j))$, so by the induction hypothesis, there is a formula $\psi$ such that $|\psi| \leq w - 1$ and  $(A(F^k/j)$, $B(*^k/j)) \models \psi$.\\
    Define $\phi:= \exists^{\geq k} x_j \psi$. We have that $A(F^k/j) \models \psi$, and so for all $1 \leq i \leq k$, $(\mathcal{A}, \alpha(F^k_{i}/j)) \models \psi$ and therefore $A \models \phi $.\\
    Now suppose that there exists a $k$-choice function $G^k$ on $B$ and a $(\mathcal{B}, \beta) \in B$  such that $\{(\mathcal{B}, \beta(G^k_i(\mathcal{B},\beta)/j)) : 1 \le i \le k\}\models \psi$. Then there is an $i \in [k]$ such that $(\mathcal{B}, \beta(G^k_i(\mathcal{B},\beta)/j))$ is part of the family $B(*^k/j)$ through the selected $G$, but $B(*^k/j) \models \neg \psi$, which is a contradiction.
    Therefore $B \models \neg \phi$.\\
    Finally $(A,B) \models \phi$ and $|\phi| = |\psi| + 1 \leq w$, so we get $(2)$.

    \item $\forall^{\geq k}$-move; a symmetrical argument to the $\bigvee$-move yields $(2)$.

    \end{itemize}
    
To prove the backward direction, assume there is a formula $\phi$ of size $w > 1$ such that $(A, B) \models \phi$. We show that Spoiler has a winning strategy in CS$^m_w(A, B)$. Let us consider the form of the formula $\phi$:

\begin{itemize}
    
    \item $\phi = \neg \psi$: Spoiler plays the $\neg$-move and gets in the position CS$^m_{w-1}(B, A)$. Since $|\psi|< w$ and $(B, A) \models \phi$, Spoiler has a winning strategy in CS$^m_{w-1}(B, A)$ by induction hypothesis, and therefore Spoiler has also a winning strategy in CS$^{m}_{w}(A,B)$.

     \item $\phi = \psi \vee \theta$: Define $C = \{(\mathcal{A}, \alpha) \in A | (\mathcal{A}, \alpha)\models \psi \}$, $D = \{(\mathcal{A}, \alpha) \in A | (\mathcal{A}, \alpha) \models \theta \}$, $u$ and $v$ such that $w = u + v$, $|\psi| \leq u $ and $|\theta| \leq v$. Spoiler plays the $\bigvee$-move associated to $u,v,C,D$ and gets to CS$^m_{u}(C, B)$ or CS$^m_{v}(D, B)$ according to Duplicator's choice. We have that $(C,B) \models \psi$, $(D,B) \models \theta$ and therefore by induction hypothesis Spoiler has a winning strategy in CS$^m_{u}(C, B)$ or CS$^m_{v}(D, B)$, and therefore Spoiler has also a winning strategy in CS$^{m}_{w}(A,B)$.

    \item $\phi = \psi \wedge \theta$: a symmetrical argument shows that Spoiler has a winning strategy in CS$^{m}_{w}(A,B)$.

    \item $\phi = \exists^{\geq k} x_j \psi$: Since $A \models \phi $, there is a $k$-choice function $F^k$ on $A$ such that $(\mathcal{A}, \alpha(F^k_{i}((\mathcal{A}, \alpha)/j)) \models \psi$, for every $1 \leq i \leq k$ and for all $(\mathcal{A}, \alpha) \in A$. Thus, $A(F^k/j) \models \psi $.\\
    On the other hand, $B \models \neg \phi $, thus from every $G^k$ on $B$ we can select $G$ such that $B(G/j) \models \neg \psi$ and therefore, $B(*^k/j) \models \neg \psi$.\\
    Finally we have that $(A(F^k/j), B(*^k/j)) \models \psi$ and $|\psi| = |\phi| - 1 < w$, so by induction hypothesis, Spoiler has a winning strategy in the game\\
    CS$^m_{w-1}(A(F^k/j), B(*^k/j))$ and therefore Spoiler has also a winning strategy in CS$^{m}_{w}(A,B)$.

    \item $\phi = \forall^{\geq k} x_j \psi$: a symmetrical argument shows that Spoiler has a winning strategy in CS$^{m}_{w}(A,B)$.

\end{itemize}
\qed

%% file: paper/6_theorems.tex
\section{Linear orders distinguinshability through Counting logics}
\label{sec:linear_order_lower_bound}

In this section, we present the new results derived from the complexity games for distinguishability and indistinguishability of finite linear orders.

A linear order is defined over the signature $\{min,max,<,succ\}$, where $<$ is a linear ordering and $succ$ is the successor relation.
$\mathcal{A}_n$ denotes the structure ($\{0, \dotsc, n\}$ , $<$), where $<$ is the standard linear order of $\{0, \dotsc, n\}$. In the rest of the paper, without loss of generality we will take linear orders on the integers, and we will denote by d$(b,a)$ the quantity $|b-a|$ for $a,b$ integers.

\subsection{Lower bound for the quantifier rank necessary to distinguish linear orders}

\begin{theorem} \label{quant_depth}
Let $t,k > 0$, and let $L_1,L_2$ be linear orders of length at least $(t+1)^k$.\\
Then $L_1 \equiv_{\mathcal{C}^{t}[k]} L_2$;\\
where $\mathcal{C}^{t}[k]$ denotes the fragment of the counting logic with counting rank at most $t$ and quantifier rank at most $k$.
\end{theorem}

 This result extends the theorem from \cite{book_libkin} stating that two linear orders of size at least $2^k$ cannot be distinguished by a $FO[k]$ formula.

The proof of this result provides a good intuition about the framework developed for the more complex Theorem \ref{theorem:result_counting}. It follows the pebble EF game of Definition \ref{pebble_game}, with bounded counting rank \cite{immerman1990describing}.

\begin{proof}
Suppose $L_1$ and $L_2$ are linear orders of length at least $(t+1)^k$, on which the EF game will be played.\\
After $i$ moves, we denote by $a$ the ``position'': $a$ is a tuple consisting of $min_{L_1},max_{L_1}$ concatenated to the $i$ moves played on $L_1$: $a = (a_{-1}, a_0, a_1,\dotsc , a_i)$, $a_{-1} = min_{L_1}$ , $a_0 = max_{L_1})$. Similarly, we define the tuple $b$ of moves played on $L_2$.\\
For $-1 \leq j, l \leq i$, we prove that regardless of Spoiler's choices, Duplicator can maintain the following inequalities:
\begin{enumerate}
    \item if $d(a_j , a_l) \leq (t+1)^{k-i}$, then $d(b_j , b_l) = d(a_j , a_l)$.
    \item if $d(a_j , a_l) > (t+1)^{k-i}$, then $d(b_j , b_l) > (t+1)^{k-i}$.
    \item $a_j \leq a_l$ iff $b_j \leq b_l.$
\end{enumerate}
Using those inequalities for $i=k$ moves, property 3. yields that Duplicator win the $k$-round EF game on $L_1$ and $L_2$. Since this happens no matter what Spoiler plays, the game characterisation theorem for bounded counting rank implies that $L_1 \equiv_{\mathcal{C}^{t}[k]} L_2$.\\
We now prove by induction on the moves (on $i$) that the inequalities can be maintained. The base case of $i = 0$ is immediate.\\
For the induction step, assume the inequalities stand for $i$ moves, and suppose without loss of generality that Spoiler makes his $(i + 1)^{\text{th}}$ move on $L_1$. Spoiler plays the set $M$ of cardinality $k$, in $L_1$. We describe Duplicator's response, a set $N$ on $L_2$. In the second part of the move, if Spoiler picks $N_q$ on $L_2$ for $q \in [t]$ (i.e. $b_{i+1} = N_q$), Duplicator will respond with $M_q$ on $L_1$ (i.e. $a_{i+1} = M_q$). \\

If $a_j = M_q$ for $j \leq i$ and $q \in [t]$, Duplicator sets $N_q$ to be $b_j$. Suppose there is an element $M_q \in M$ for which this is not the case. We define $j,l \leq i$ such that $a_j < M_{q} < a_l$ and that there is no other previously played moves on $L_1$ inside this interval. By property 3., the interval between $b_j$ and $b_l$ contains no other elements of $b$. Then we have two cases regarding the length of the interval:

\begin{itemize}
    \item If $d(a_j , a_l) \leq (t+1)^{k-i}$. Then by property 1. and inductive hypothesis, $d(b_j , b_l) = d(a_j , a_l)$, and the intervals $[a_j , a_l]$ and $[b_j , b_l]$ are isomorphic. Duplicator picks $N_{q}$ so that $d(a_j , M_{q}) = d(b_j , N_{q})$ and $d(M_{q}, a_l) = d(M_{q}, b_l)$, which ensures that the three properties hold for $i+1$ moves.
    \item  $d(a_j , a_l) > (t+1)^{k-i}$. In this case by property 2., $d(b_j , b_l) > (t+1)^{k-i}$. We have three possibilities: \begin{itemize}
        \item If $d(a_j , M_{q}) \leq (t+1)^{k-(i+1)}$. Then $d(M_{q}, a_l) > (t+1)^{k-(i+1)}$, and Duplicator picks $N_{q}$ on $[b_j,b_l]$ so that $d(b_j , N_{q}) = d(a_j , M_{q})$ maintaining properties 1 and 3. Since $d(M_{q}, b_l) > (t+1)^{k-(i+1)}$ this maintains $d(N_{q}, b_l) > (t+1)^{k-(i+1)}$ by property 2.$(i)$, hence property 2.$(i+1)$.
    \item $d(M_{q}, a_l) \leq (t+1)^{k-(i+1)}$, in which case a similar reasoning applies.

    \item Otherwise, Spoiler has picked $M_q$ such that both: $d(a_j , M_q ) > (t+1)^{k-(i+1)} $, $\, d(M_q , a_l ) > (t+1)^{k-(i+1)}$. There can be at most $t$ such $M_q$ as $|M| = t$.\\
    Since $d(b_j , b_l) > (t+1)^{k-i}$, there are at least $t$ distinct elements in $L_2$ between $b_j+ (t+1)^{k-(i+1)}$ and $b_l - (t+1)^{k-(i+1)}$. Duplicator picks the first $N_q$ on that interval that is not already in $N$. This ensures that, $d(b_j , N_q) > (t+1)^{k-(i+1)}$ and $d(N_q, b_l) > (t+1)^{k-(i+1)}$, therefore satisfying all three properties.
    \end{itemize}

\end{itemize}
Thus, in all the cases, the induction is preserved.
\end{proof}

\subsection{Upper bound for the minimal size necessary to distinguish linear orders}
The following proposition gives an upper bound on the size of ${\mathcal C}_2^t$- formulas distinguishing
linear orders of different sizes.
\begin{proposition}
\label{proposition:upper_bound}
There is a ${\mathcal C}^t_2$-formula $\phi$ of size $O(n/t)$ such that for $n < m$ it holds that $(\mathcal{A}_n, \mathcal{A}_m ) \models \phi$.
\end{proposition}

  \begin{proof}
  Consider the formulae:
\begin{enumerate}
    \item $\phi_{0}(x) = (x=x)$.
    \item $\phi_{t (l+1)}(x) = \exists^{\geq t} y ((y<x) \wedge \phi_{t l}(y))$.
    \item For $n=tl + p$ with $p<t$: $\phi_{tl + p }(x) = \exists^{\ge p} y ((y<x) \wedge \phi_{t l}(y))$
\end{enumerate}
Let $\phi = \neg \exists x \,\, \phi_{n+1}(x)$
Then a linear order satisfies $\phi$ iff it has size at most $n$.

The size of  $\phi$ is $O(n/t)$.
\end{proof}

\subsection{Lower bound for the size necessary of a $\mathcal{C}^t_3$ formula to distinguish linear orders}
Using the CS$^{3,t}$ game introduced in Section \ref{sec:gchar}, we derive the following lower bound for the size of a $\mathcal{C}^t_3$ formula distinguishing linear orders:

\begin{theorem}\label{theorem:result_counting}
If $\phi$ is a $\mathcal{C}_3^t$-formula distinguishing $\mathcal{A}_n$ and $\mathcal{A}_m$,
where $n < m$, then $|\phi| \ge  \frac{\sqrt{n}}{t}$.     
\end{theorem}

For $t=1$, i.e. for FO$_3$, this theorem improves on the best known lower bound of \cite{grohe_succinc} from $\sqrt{n}/2$ to $\sqrt{n}$.\\
The rest of the paper is dedicated to proving this result.

%% file: paper/7_full_proof.tex
\section{Proof of Theorem \ref{theorem:result_counting}}
\label{section:proof}
\subsection{Description of the framework}
We extend the framework for FO of ~\cite{grohe_succinc} to counting logic.
We start by describing the main concepts: the counting game derived extended syntax trees, the separators, and the weighting scheme.
\subsubsection{Extended Syntax Tree}
Given two families $A,B$, the extended syntax tree represents a winning strategy for Spoiler in the CS$^{3,t}$ game on $(A,B)$. It assigns to each tree node $v$ a pair of families $il(v)$ (``interpretation label''), along with Spoiler's move $sl(v)$ (``syntax label'') in this position such that:
\begin{enumerate}
    \item A node $v_1$ is a child of node $v$ if position $il(v_1)$ can be obtained in one move from $il(v)$,
    \item A node $v$ is a leaf if $il(v)$ satisfies the atomic win condition. 
\end{enumerate}
The root is associated to the starting position of the game CS$^{3,t}$, and we consider the nodes associated to positions reached through a winning strategy.
\begin{definition}
\label{definition:syntax_tree}
Let $\psi$ be an $\mathcal{C}^t_3$-formula, let $A$ and $B$ be families of interpretations such that $(A,B) \models \psi$. By induction on the construction of $\psi$, we define an extended syntax tree $T_{\psi}^{\langle A,B \rangle}$ as follows:
\begin{itemize}
    \item If $\psi$ is an atomic formula, then $T_{\psi}^{\langle A,B \rangle}$ consists of a single node $v$ that has a syntax label
$sl(v) := \psi$ and an interpretation label $il(v) := \langle A,B \rangle$.
    \item If $\psi$ is of the form $ \neg \psi_1$, then $T_{\psi}^{\langle A,B \rangle}$ has a root node $v$ with $sl(v) := \neg$ and $il(v) := \langle A,B \rangle$. The unique child of $v$ is the root of $T_{\psi_1}^{\langle B,A \rangle}$. Note that $(B,A) \models  \psi_1$.
    
    \item If $\psi$ is of the form $\psi_1 \vee \psi_2$, then $T_{\psi}^{\langle A,B \rangle}$ has a root node $v$ with $sl(v) := \vee$ and $il(v) := \langle A,B \rangle$. The first child of $v$ is the root of $T_{\psi_1}^{\langle A_1,B \rangle}$ and the second child of $v$ is the root of $T_{\psi_2}^{\langle A_2,B \rangle}$, where $A_i = \{(\mathcal{A}, \alpha) \in A : (\mathcal{A}, \alpha) \models \psi_i\}$ for $i \in \{1, 2\}$. Note that $A= A_1 \cup A_2$ and $(A_i,B) \models \psi_i$.
    
    \item If $\psi$ is of the form $\psi_1 \wedge \psi_2$, then $T_{\psi}^{\langle A,B \rangle}$ has a root node $v$ with $sl(v) := \wedge$ and $il(v) := \langle A,B \rangle$. The first child of $v$ is the root of $T_{\psi_1}^{\langle A,B_1 \rangle}$ and the second child of $v$ is the root of $T_{\psi_2}^{\langle A,B_2 \rangle}$, where $B_i = \{(\mathcal{B}, \beta) \in B : (\mathcal{B}, \beta) \models \neg \psi_i\}$ for $i \in \{1, 2\}$. Note that $B = B_1\cup B_2$, $(A,B_i) \models \psi_i$.

    \item If $\psi$ is of the form $\exists^{\geq k} u \psi_1$, for a variable $u \in \{x, y, z\}$ and $k \leq t$, then $T_{\psi}^{\langle A,B \rangle}$ has a root node $v$ with $sl(v) := \exists^{\geq k} u$ and $il(v) := \langle A,B \rangle$. The unique child of $v$ is the root of $T_{\psi_1}^{\langle A(F^k/u),B(*^k/u) \rangle}$, where $F^k \in F^k_A$ is chosen so that $A(F^k/u) \models \psi_1$, and for every $G^k \in F^k_B$, $G$ is selected from $G^k$ so that $B(G/u) \models \neg \psi_1$. Since $B(*^k/u) = \bigcup_{G^k \in F^k_B} B(G/u) $, note that \\
    $(A(F^k/u),B(*^k/u)) \models \psi_1$.

    \item If $\psi$ is of the form $\forall^{\geq k} u \psi_1$, for a variable $u \in \{x, y, z\}$ and $k \leq t$, then $T_{\psi}^{\langle A,B \rangle}$ has a root node $v$ with $sl(v) := \forall^{\geq k} u$ and $il(v) := \langle A,B \rangle$. The unique child of $v$ is the root of $T_{ \psi_1}^{\langle A(*^k/u),B(F^k/j) \rangle}$, where $F^k \in F^k_B$ is chosen so that $B(F^k/u) \models \neg \psi_1$, and for every $G^k \in F^k_A$, $G$ is selected from $G^k$ so that $A(G/u) \models \psi_1$. Since $A(*^k/u) = \bigcup_{G^k \in F^k_A} A(G/u) $, note that $(A(*^k/u),B(F^k/u)) \models \psi_1$.

\end{itemize}
\end{definition}
We define $|T_{\psi}^{\langle A,B \rangle}|$ to be the number of nodes in the tree $T_{\psi}^{\langle A,B \rangle}$. It is trivial to check by induction that $|T_{\psi}^{\langle A,B \rangle}| = |\psi|$.\\

Next we specify separators, the key concept to study succinctness on linear orders. Separators express the distance from which, at a given position in the game (i.e. at a certain size in the formula), we are able to distinguish two elements of linear orders.
This distance differs for each element of $\mathcal{P}_2 (\{min, max, x, y, z\})$, where the $\mathcal{P}_2$ operator describes the subsets of size two.

\subsubsection{Separators and Weighting Scheme}

\begin{definition}[separator]
Let $A$ and $B$ be sets of interpretations. A separator is a mapping $\delta : \mathcal{P}_2 (\{min, max, x, y, z\}) \to \mathbb{N} $. $\delta$ is called separator for $\langle A,B \rangle$, if the following is satisfied:\\

For every $I := (\mathcal{A}, \alpha) \in A$ and $J := (\mathcal{B}, \beta) \in B $, there are $u, u' \in \{min, max, x, y, z\}$ with $u \neq u'$, such that: 

\begin{enumerate}
    \item $<$-type$(\alpha(u), \alpha(u')) \neq $$<$-type$(\beta(u), \beta(u'))$  or
    \item both:
    \begin{itemize}
        \item MIN$[\text{d}(\alpha(u), \alpha(u' )) , \text{d}(\beta(u), \beta(u'))] \leq \delta(\{u, u'\})$ and,
        \item $\text{d}(\alpha(u), \alpha(u' )) \neq \text{d}(\beta(u), \beta(u' ))$.
    \end{itemize}
\end{enumerate}

\noindent where, for $a,b \in \mathbb{N}$, d$(a,b) := |a-b|$ and $<$-type$(a,b)$ is $"=", "<" $ or $ ">"$ reflecting the order between $a$ and $b$.
\end{definition}

\begin{definition}[weight of a separator and minimal separator]
\label{definition:weight}
Let $\delta$ be a separator, we define its:
\begin{enumerate}
    \item border-distance $b(\delta) :=$ MAX $\{ \delta(\{min, max\}),$ \\ $\delta(\{min, u\}) + \delta(\{u', max\}) : u, u' \in \{x, y, z\} \}$
    \item centre-distance
$c(\delta) :=$ MAX $  \{\delta(p) + \delta(q) : p, q \in \mathcal{P}_2(\{x, y, z\}), p \neq q \}$
    \item weight $w(\delta) := \sqrt{c(\delta)^2 + b(\delta)}$.
\end{enumerate}
\end{definition}

The weight is a measure of the distinguishing power of separators. We call \textit{minimal separator} of a pair of structures a separator of minimal weight.
\subsection{Proof of Theorem~\ref{theorem:result_counting}}

We introduce the main technical lemma of the proof, which links separators weights to the structure of the tree.
\begin{lemma}
\label{lemma:separator}
Suppose $(A,B) \models \psi$ and let $T$ be the extended syntax tree $T_{\psi}^{\langle A,B \rangle}$.
For every node $v$ of $T$ the following is true:
\begin{enumerate}
    \item If $v$ is a leaf, then $w(\delta) \leq 1$.
    \item If $v$ has 2 children $v_1$ and $v_2$, and $\delta_i$ is a minimal separator for $il(v_i)$, for $i \in \{1, 2\}$, then $w(\delta) \leq w(\delta_1) + w(\delta_2)$.
    \item If $v$ has exactly one child $v_1$, and $\delta_1$ is a minimal separator for $il(v_1)$, then $w(\delta) \leq w(\delta_1) + t$.
\end{enumerate}
where $\delta$ is a minimal separator for $il(v)$.
\end{lemma}

The proofs of parts 1. and 2., related only to the syntax of first-order logic, are almost the same in our counting logic extended syntax tree framework as in~\cite{grohe_succinc}, and are given in Appendix \ref{appendix:lemma_parts_1_2}.
However, the proof of Lemma \ref{lemma:separator} part 3., which accounts for the counting quantifiers, 
requires different arguments, to take into account the ``multiple choices'' needed in counting logic, and is given in Section~\ref{sec:sep}.\\
The bound of part 3. extends the corresponding bound of~\cite{grohe_succinc} to counting logic. In fact,
for counting rank $1$, i.e., for first-order logic without counting, it improves the bound  from $w(\delta) \leq w(\delta_1) + 2$ to $w(\delta) \leq w(\delta_1) + 1$. This leads to the slight improvement of the lower bound of~\cite{grohe_succinc} mentioned after Theorem \ref{theorem:result_counting}.

From Lemma \ref{lemma:separator}'s inductive relation on the minimal weight of extended syntax tree nodes, we derive a lower bound on the size of the tree in terms of the minimal weight of the root.

\begin{lemma}
\label{lemma:tree}
 Let $T$ be a finite binary tree where each node $v$ is equipped with a weight
$w(v) > 0$ such that the following is true:
\begin{itemize}
    \item If v is a leaf, then $w(v) \leq 1$.
    \item If $v$ has 2 children $v_1$ and $v_2$, then
$w(v) \leq w(v_1) + w(v_2)$.
    \item If $v$ has exactly one child $v_1$, then $w(v) \leq w(v_1) + t$.
\end{itemize}
Then, $|T| \geq \frac{w(r)}{t}$, where $r$ is the root of $T$ and $|T|$ is the number of nodes of $T$.
\end{lemma}
\begin{proof}
The proof is by induction:
\begin{itemize}
    \item If $T$ is a leaf,  $|T| = 1 \geq \frac{w(r)}{t}$.
    \item If $T$'s root starts with two children $v_1, v_2$ associated to the trees $T_1,T_2$. By inductive hypothesis $|T_i| \geq \frac{w(v_i)}{t}$, so $|T| = |T_1| +  |T_2| + 1 \geq \frac{w(v_1)}{t} + \frac{w(v_2)}{t} + 1 \geq \frac{w(r)}{t} $.
    \item Finally if $T$'s root starts with one child $v_1$ associated to the tree $T_1$. By inductive hypothesis $|T_1| \geq \frac{w(v_1)}{t}$, so $|T| = |T_1| + 1 \geq \frac{w(v_1)}{t} + \frac{t}{t} \geq \frac{w(r)}{t} $.
\end{itemize}
\end{proof}

Finally we can derive a lower bound for the size of the extended syntax trees on the counting logic with the main result:
\begin{theorem} \label{theorem:main}
Suppose $(A,B) \models \psi$ and $\delta$ is a minimal separator for $\langle A,B \rangle$, then $|\psi| \geq \frac{w(\delta)}{t}$.     
\end{theorem}

\begin{proof}
Consider $T_{\psi}^{\langle A,B \rangle}$ the syntax tree for the pair $\langle A,B \rangle$, $\psi$.\\
We associate to each node $v$ of $T_{\psi}^{\langle A,B \rangle}$ a weight $w(v) := w(\delta_v)$, where $\delta_v$ is a minimal separator
for $il(v)$. From Lemmas \ref{lemma:separator} and \ref{lemma:tree}, we get that $|\psi| = |T_{\psi}^{\langle A,B \rangle}| \geq \frac{w(\delta)}{t}$, where $\delta$ is a minimal separator of $il(r) = <A, B>$.
\end{proof}

Finally, as a consequence, we get Theorem \ref{theorem:result_counting}.
\begin{proof}
Suppose $\psi$ is an $\mathcal{C}^t_3$-sentence such that
$\mathcal{A}_n \models\psi$ and $\mathcal{A}_m \models \neg \psi$. Let $\alpha$ be the assignment that assigns $x,y$ and $z$ to $0$. Consider the separator $\delta : \mathcal{P}_2 (\{min, max, x, y, z\}) \to \mathbb{N} $ defined as:
\begin{itemize}
    \item $\delta(\{min,max\}) = n$, 
    \item $\delta(\{u,u'\}) = 0$ for $(u,u') \in  \mathcal{P}_2 (\{min, max, x, y, z\}) \setminus \{(min,max)\}$.
\end{itemize}
It is easy to see that $w(\delta) = \sqrt{n}$ and that $\delta$ is a minimal separator for $\langle (\mathcal{A}_n,\alpha), (\mathcal{A}_m,\alpha) \rangle$. Then Theorem \ref{theorem:result_counting} follows from Theorem \ref{theorem:main} and this observation.

\end{proof}

\begin{figure*}[h!]
    \centering
    \includegraphics[width=0.8\linewidth]{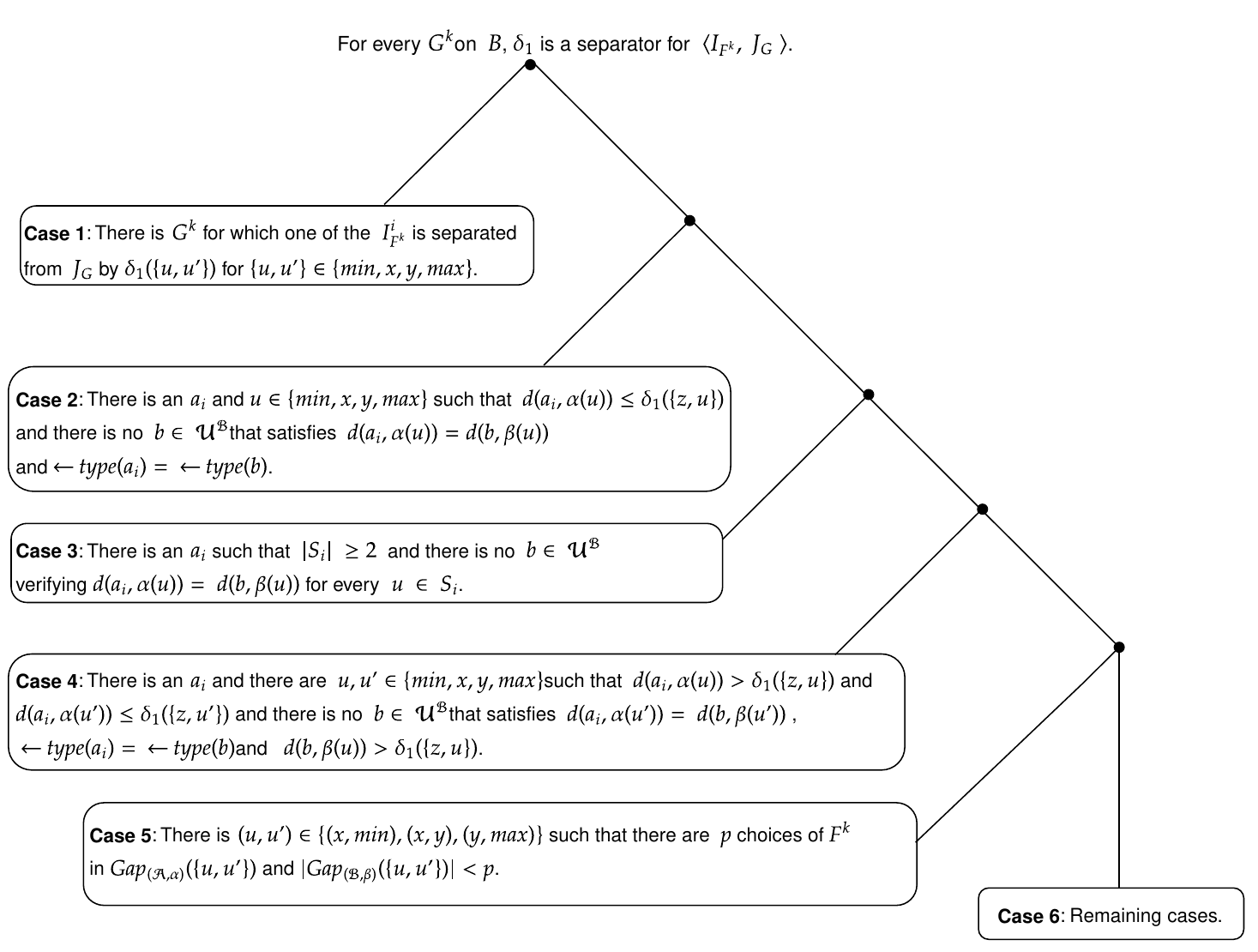}
    \caption{The structure of the proof.}
    \label{fig:proof_structure}
\end{figure*}
\subsection{Proof of Lemma \ref{lemma:separator}} 
\label{sec:sep}
The proof of Lemma \ref{lemma:separator} part 3. relies on the following key lemma:

\begin{lemma}\label{key:Lemma}Let $v$ be a node of $T$ that has syntax-label $sl(v) = Q^{\geq k} u$ for $Q \in \{\exists, \forall\}$, $k \leq t$ and $u \in \{x, y, z\}$. Let $\delta_1$ be a separator for $il(v_1)$, where $v_1$ is the unique child of $v$ in $T$. Let $\delta$ be the separator defined via:
\begin{itemize}
    \item $\delta(\{u, u'\}) := 0$ , for all $u' \in \{min, max, x, y, z\} \setminus \{u\}$,
    
    \item 
 $\delta(\{min, max\}) := $MAX$ \{ \delta_1(\{min, max\}), \delta_1(\{min, u\}) \\+ \delta_1(\{u, max\}) + k-1  \}$.
 
    And for all $u', u''$ such that $\{x, y, z\} = \{u, u', u''\}$ and all $m \in \{min, max\}$:
    \item $\delta(\{u', u''\}) := $MAX$ \{ \delta_1(\{u', u''\}), \delta_1(\{u', u\}) $\\
    $+ \delta_1(\{u, u''\}) + k-1  \}$,
    
    \item $\delta(\{m, u'\}) := $MAX$\{ \delta_1(\{m, u'\}) , \delta_1(\{m, u\}) $ \\ $+ \delta_1(\{u, u'\}) + k-1 \}$ .
\end{itemize}
Then,  $\delta$ is a separator for $il(v)$. 
\end{lemma}
Lemma~\ref{key:Lemma} precisely captures how counting quantifiers increase the distinguishing power of separators. From Lemma \ref{key:Lemma}, we obtain Lemma \ref{lemma:separator} part 3. by building inequalities on the weight, which is done in Appendix \ref{appendix:consequence_weight}.

\smallskip

\subsubsection{Overview of the proof of Lemma \ref{key:Lemma}}

\begin{proof}
Due to symmetry, we only consider the case $\exists^{\geq k} z$.
It has to be shown that $\delta$ is a separator for $\langle A, B \rangle = il(v)$. By definition,
$il(v_1) = \langle A(F^k/z),B(*^k/z) \rangle$.
Consider $I = (\mathcal{A}, \alpha) \in A $ and $J =  (\mathcal{B}, \beta) \in B $. 
Let:
\begin{itemize}

    \item $I_{F^k} := \{(\mathcal{A}, \alpha(F^k_i(\mathcal{A},\alpha)/z)\}$ be the set of interpretations generated from $I$ in $v_1$, and let $a_i = F^k_i(\mathcal{A},\alpha)$ and $ I^i_{F^k} = ({\mathcal A},\alpha(a_i/z))$, $1 \leq i \leq k$.

    \item For any $k$-choice function $G^k$ on $B$, let $G$ be the choice function selected from $G^k$, and let $J_{G}$ be the interpretation selected from $J_{G^k} = \{(\mathcal{B}, \beta(G^k_i(\mathcal{B}, \beta)/z))\}$. We define
     $b_i := G^k_i(\mathcal{B},\beta)$; note that $J_{G}$ corresponds to some $b_j$,  $1 \leq j \leq k$.

\end{itemize}

For every $G^k$ on $B$, $\delta_1$ is a separator for $\langle I_{F^k},J_{G} \rangle$. 
Thus for every $a_i$ there is a pair $u_{G}^i, v_{G}^i$ in $\{min, max, x, y, z\}$ such that $\langle I_{F^k}^i,J_{G} \rangle$ is separated by $\delta_1(\{u_{G}^i, v_{G}^i\})$. 
The proof of the Lemma follows by considering several cases, as illustrated in Fig. ~\ref{fig:proof_structure}. In each case, it is shown that $\delta$ separates $\langle I,J\rangle$. {\bf Cases 1-5} are \emph{overlapping}, and after considering these cases we show that they cover all the possibilities, i.e. {\bf Case 6} leads to a contradiction.\\
The first case, when $z$ is not needed in $\delta_1$, is the same as the first case in~\cite{grohe_succinc}. {\bf Cases 2-4} show that if there is an $a_i$ in some \emph{nearness relationship} with the interpretation $(\mathcal{A}, \alpha)$ according to $\delta_1$, then $\delta$ is a separator. 
In {\bf Case 5}, we introduce the concept of \emph{gap}, which is the set of elements not covered by the nearness conditions imposed by $\delta_1$. The condition of {\bf Case 5} depends on the \emph{number} of elements
$a_i$ in the gaps. This is the point where the additive term $k$ enters the bound. Finally, in {\bf Case 6} is shown that these cases cover all possibilities, as otherwise there is a set $\{b_1, \ldots, b_k\}$ duplicating $\{a_1, \ldots, a_k\}$ in the sense that no $b_j$ allows  $\delta_1$ to separate.
{\bf Case 6} in particular explains the improvement of our proof over the result of \cite{grohe_succinc} for FO, as a ``$+1$'' is introduced in \cite{grohe_succinc} to deal with what would encompass that case.

\begin{figure*}[h]
    \centering
    \includegraphics[width=0.8\linewidth]{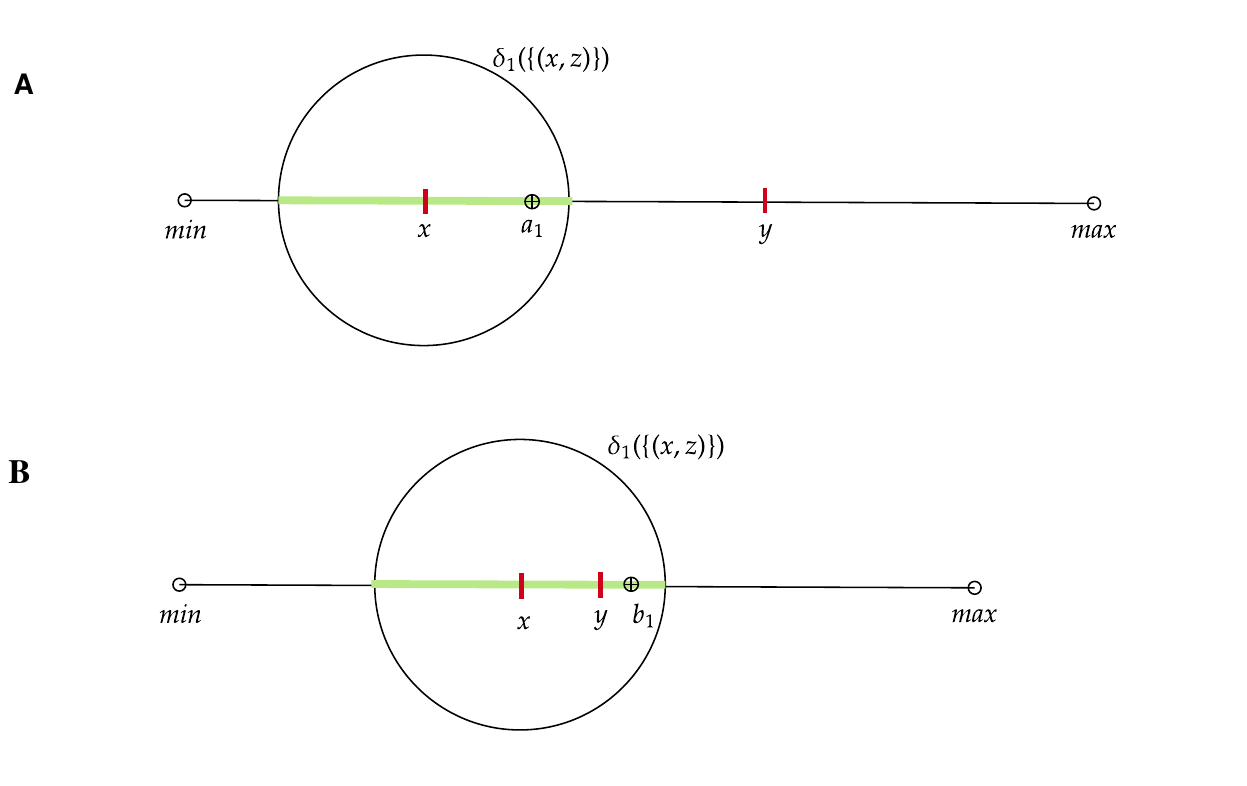}
    \caption{{\bf Case 2} of the proof of Lemma \ref{key:Lemma}.\\
    We associate a ball with center $\alpha(u)$ of radius $\delta_1({u,v})$ for every $u \in \{x,y,min,max\}$ 
on $\mathcal{A}$ (resp. $\beta(u)$ on $\mathcal{B}$). 
Each ball associated to variable $u$ accounts for the ability of the pair $\{u,z\}$ to separate. The crossed dot on $\mathcal{A}$ represents a choice on the linear order from $F^k$, the crossed dot on $\mathcal{B}$ represents ``a best attempt'' at a response for indistinguishability. In {\bf Case 2}, the choice on $\mathcal{A}$ is in a nearness relationship (for $\delta_1(\{x,z\})$ on the figure), and there is no element of $\mathcal{B}$ that matches the separating conditions for indistinguishability (in the figure, d$(a_1, \alpha(x)) =$ d$(b_1, \beta(x)) $ but $<$-type$(a_1) \neq$ $<$-type$(b_1)$).}
    \label{fig:case2}
\end{figure*}

\smallskip

\subsubsection{Rest of the proof}

Note that for all $u \in \{min, max, \\ x, y\}$, and for all $i \in [k]$,  $\alpha(F^k_i(\mathcal{A}, \alpha)/z)(u)=\alpha(u)$. So we will
omit subscripts and just write $\alpha(x)$ for the assignment of $x$ and $min_A$ for the minimum in the $k$ copies of $(\mathcal{A}, \alpha)$ (and similarly with $y,max$). Similarly, for all $u \in \{min, max, x, y\}$, for all $ i \in [k]$, and for all $G^k \in F^k_B$,\\ $\beta(G^k_i(\mathcal{B}, \beta)/z)(u)=\beta(u)$, so we will also write $\beta(x)$ for the assignment of $x$ and $min_B$ for the minimum in the $k$ copies of $(\mathcal{B}, \beta)$ (and similarly with $y,max$). We will say that $\delta(\{u,u'\})$ separates $\langle I,J \rangle$ to signify that the pair of variable $\{u,u'\}$ witnesses the separation property of $\delta$ on $\langle I,J \rangle$.

\smallskip
{\bf Case 1}: There is $G^k$ for which one of the $I^i_{F^k}$ is separated from $J_{G}$ by $\delta_1(\{u, u'\})$ for $\{u,u'\} \in \{min, x,y, max\}$.

\smallskip 

 If there exists $\{u,u'\} \in \{min, x,y, max\}$ such that $\delta_1(\{u, u'\})$ separates $\langle I_{F^k},J_{G} \rangle$ then by definition $\delta(\{u, u'\})$ separates $\langle I,J \rangle$. \qed

\smallskip 
{\bf Assumption 1} In the rest of the proof we will suppose that none of the $I^i_{F^k}$ is separated from $J_{G}$ by $\delta_1(\{u, u'\})$ for $\{u,u'\} \in \{min, x,y, max\}$.\\
Without loss of generality, we may assume that $\alpha(x) \leq \alpha(y)$. If  $\beta(x) > \beta(y)$, $\delta_1(\{x,y\})$ separates $\langle I_{F^k},J_{G} \rangle$ for every $G^k$, which is a contradiction, so $\beta(x) \leq \beta(y)$.

\smallskip

To be able to consider the $k$ choices simultaneously, we are looking to find for each choice $a_i$, a $b_i \in \mathcal{U}^\mathcal{B}$ that $\delta_1$ will not be able to distinguish, in an adversary mindset. An element $a_i$ has a particular location in $(\mathcal{A},\alpha)$ with regards to  $\delta_1(\{z,u\})$ for $u \in \{x,y,min,max\}$, namely that either d$(a_i, \alpha(u)) \leq \delta_1(\{z,u\})$ or d$(a_i, \alpha(u)) >\delta_1(\{z,u\})$. An $a_i$ location is also based on its type ($"=", "<", ">"$) relative to $min_A,\alpha(x), \alpha(y), max_A$, which we will refer to as $<$-type$(a_i)$ ($<$-type$(a_i)$ can be seen as the 4-tuple ($<$-type$(a_i,min_A)$, $<$-type$(a_i,\alpha(x))$, $<$-type$(a_i,\alpha(y))$, $<$-type$(a_i,max_A))$).

\smallskip
{\bf Case 2}: There is an $a_i$ and a $u \in \{x,y,min,max\}$ such that d$(a_i, \alpha(u)) \leq \delta_1(\{z,u\})$ and there is no $b \in \mathcal{U}^{\mathcal{B}}$ that satisfies d$(a_i, \alpha(u)) =$ d$(b, \beta(u)) $ and $<$-type$(a_i) =$ $<$-type$(b)$.

\smallskip

\begin{figure*}[h]
    \centering
    \includegraphics[width=0.8\linewidth]{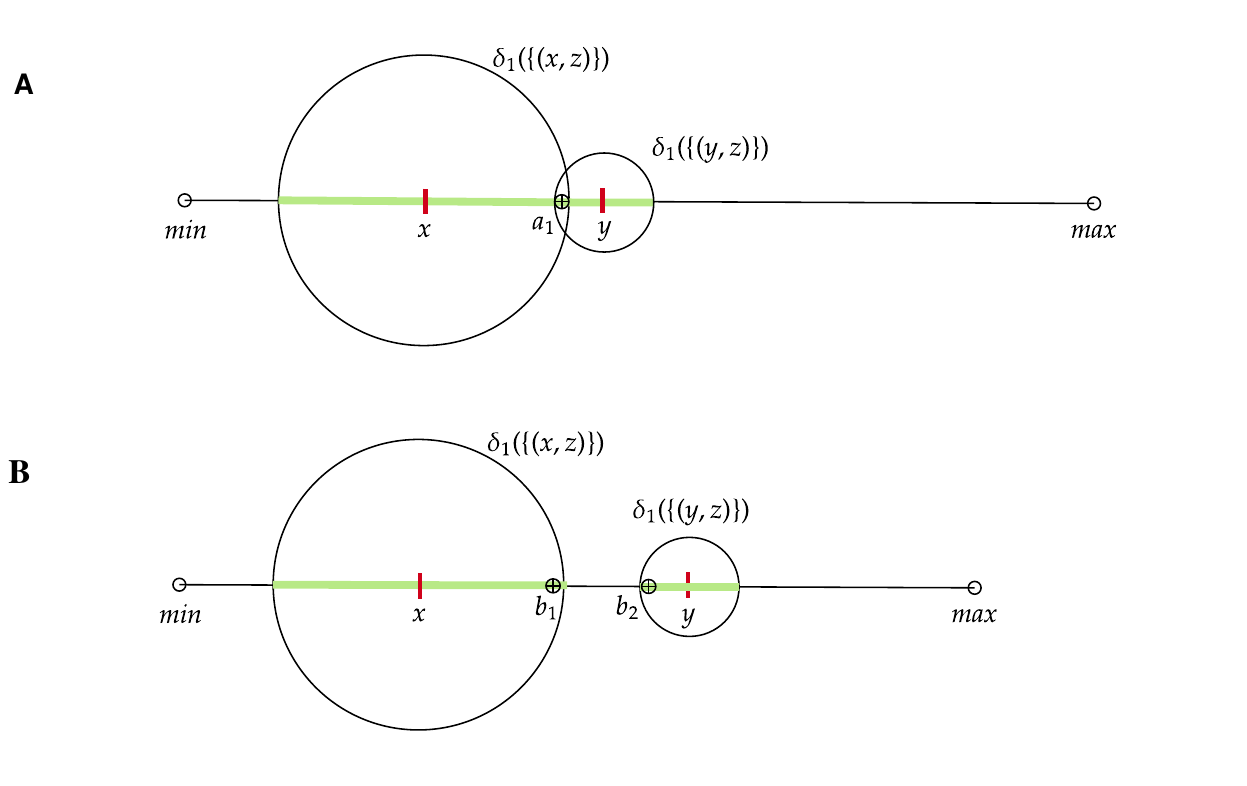}
    \caption{Case 3.}
    \label{fig:case3}
\end{figure*}

We consider two sub-cases to better handle several variables at the same time. {\bf Case 2.2} is represented in Fig. \ref{fig:case2}.

\smallskip
{\bf Case 2.1}: There is an $a_i$ and a $u \in \{x,y,min,max\}$ such that d$(a_i, \alpha(u)) \leq \delta_1(\{z,u\})$ and there is no $b \in \mathcal{U}^{\mathcal{B}}$ such that d$(a_i, \alpha(u)) =$ d$(b, \beta(u)) $ and $<$-type$(a_i,\alpha(u)) = $ $<$-type$(b,\beta(u))$.

\smallskip

We have four cases:
\begin{itemize}
    \item [-]$u$ is $min$: then 
    then $\delta(\{min, max\}) \geq \delta_1(\{min, z\}) \geq $ d$(a_i, min_A)$ $ > $ d$(max_B, min_B)$. Therefore:
    \begin{itemize}
        \item [$\bullet$] d$(max_B, min_B) \neq $ d$(max_A, min_A)$;

        \item [$\bullet$] d$(max_B, min_B) \leq \delta(\{min, max\})$.
    \end{itemize}
    So $\delta(\{min,max\})$ separates $\langle I,J \rangle$.
    \item [-] $u$ is $max$: the above reasoning applies by symmetry.
    \item  [-] $u$ is $x$: If $\alpha(x) < a_i$, 
$\delta(\{x, max\}) \geq \delta_1(\{x, z\}) \geq $ d$(a_i, \alpha(x))$ $ > $ d$(max_B, \beta(x))$. Then:
    \begin{itemize}
        \item [$\bullet$] d$ (max_A, \alpha(x)) \neq$ d$(max_B, \beta(x))$;

        \item [$\bullet$] d$(max_B, \beta(x)) \leq \delta(\{x, max\}) $.
    \end{itemize}
So $\delta(\{x, max\})$ separates $\langle I,J \rangle$. If $\alpha(x) > a_i$, we reach a similar conclusion. Note that $\alpha(x)=a_i$ contradicts the assumption.
    \item  [-] $u$ is $y$: the above reasoning applies by symmetry.\qed
\end{itemize}
\smallskip
{\bf Case 2.2}: There is an $a_i$ and a $u \in \{x,y,min,max\}$ such that d$(a_i, \alpha(u)) \leq \delta_1(\{z,u\})$ and there is no $b \in \mathcal{U}^{\mathcal{B}}$ such that d$(a_i, \alpha(u)) =$ d$(b, \beta(u)) $ and $<$-type$(a_i) =$ $<$-type$(b)$.

\smallskip

We assume we are not in {\bf Case 2.1}, so there exists a $b \in \mathcal{U}^{\mathcal{B}}$ verifying d$(a_i, \alpha(u)) =$ d$(b, \beta(u)) $ and $<$-type$(a_i,\alpha(u)) = $ $<$-type$(b,\beta(u))$, note that such a $b$ is unique, as we are working on linear orders. We fix $b$ and call $u'$ the variable in $\{x,y,min,max\}$ such that $<$-type$(a_i,\alpha(u')) \neq $ $<$-type$(b,\beta(u'))$ by hypothesis of {\bf Case 2.2}.
\noindent
 The type is different, so it must be that both:
 \begin{itemize}
 \item d$(\alpha(u'), \alpha(u)) \neq $ d$(\beta(u'), \beta(u))$;
\item  MIN$[$d$(\alpha(u'), \alpha(u)),$ d$(\beta(u'), \beta(u)) ]\leq $ d$(a_i, \alpha(u))$.
\end{itemize}
Since d$(a_i, \alpha(u)) \leq \delta_1(\{z,u\})$ $ \leq \delta(\{u',u\})$ we get:
    \begin{itemize}
        \item d$(\alpha(u'), \alpha(u)) \neq $ d$(\beta(u'), \beta(u))$;
        \item MIN$[$d$(\alpha(u'), \alpha(u)),$ d$(\beta(u'), \beta(u)) ] \leq \delta(\{u',u\}) $.
    \end{itemize}
So $\delta(\{u',u\})$ separates$\langle I,J \rangle$.\qed

\smallskip
{\bf Assumption 2} In the rest of the proof, we will suppose that if there is an $a_i$ and a $u\in \{x,y,min,max\}$ such that d$(a_i, \alpha(u)) \leq \delta_1(\{z,u\})$, then there is a $b \in \mathcal{U}^{\mathcal{B}}$ such that d$(a_i, \alpha(u)) =$ d$(b, \beta(u)) $ and $<$-type$(a_i) =$ $<$-type$(b)$.

\smallskip

We define $S_i :=\{ u \in \{x,y,min,max\} \, : \, $d$(a_i, \alpha(u)) \leq \delta_1(\{z,u\})\}$, the set of variables to which $a_i$ is ``near''.

\smallskip

{\bf Case 3}: There is an $a_i$ such that $|S_i| \geq 2$ and there is no $b \in \mathcal{U}^{\mathcal{B}}$ verifying d$(a_i, \alpha(u)) =$ d$(b, \beta(u)) $ for every $u \in S_i$.

We give a visual representation of this case in Fig. \ref{fig:case3} in Appendix \ref{appendix:figures}.
\smallskip

   Suppose there is no such $b$, we call $u,u'$ two elements of $S_i$. Then {\bf Assumption 2} implies that d$(\alpha(u), \alpha(u')) \neq $ d$(\beta(u), \beta(u'))$. Notice that d$(\alpha(u), \alpha(u')) \leq $ d$(a_i, \alpha(u)) + $\\ d$(a_i, \alpha(u'))$ $  \leq \delta_1(\{z,u\}) + \delta_1(\{z,u'\}) \leq \delta(\{u,u'\})$. So:
   \begin{itemize}
       \item  d$(\alpha(u), \alpha(u')) \neq $ d$(\beta(u), \beta(u'))$;
       \item d$(\alpha(u), \alpha(u')) \leq \delta(\{u,u'\})$.
   \end{itemize}
   Therefore $\delta(\{u,u'\})$ separates $\langle I,J \rangle$.\qed
   
\smallskip
{\bf Assumption 3} In the rest of the proof, we will suppose that if there is an $a_i$ such that $|S_i| \geq 2$, then there is $b \in \mathcal{U}^{\mathcal{B}}$ such that d$(a_i, \alpha(u)) =$ d$(b, \beta(u)) $ for all $u \in S_i$.

\smallskip
{\bf Case 4}: There is an $a_i$ and there are $u,u' \in \{x,y,min,max\}$ such that d$(a_i, \alpha(u)) > \delta_1(\{z,u\})$ and d$(a_i, \alpha(u')) \leq \delta_1(\{z,u'\})$ and there is no $b \in \mathcal{U}^{\mathcal{B}}$ that satisfies d$(a_i, \alpha(u'))=$ d$(b, \beta(u'))$,  $<$-type$(a_i) =$ $<$-type$(b)$ and d$(b, \beta(u)) > \delta_1(\{z,u\})$.

We give a visual representation of this case in Fig. \ref{fig:case4}.

\begin{figure*}[h]
    \centering
    \includegraphics[width=0.8\linewidth]{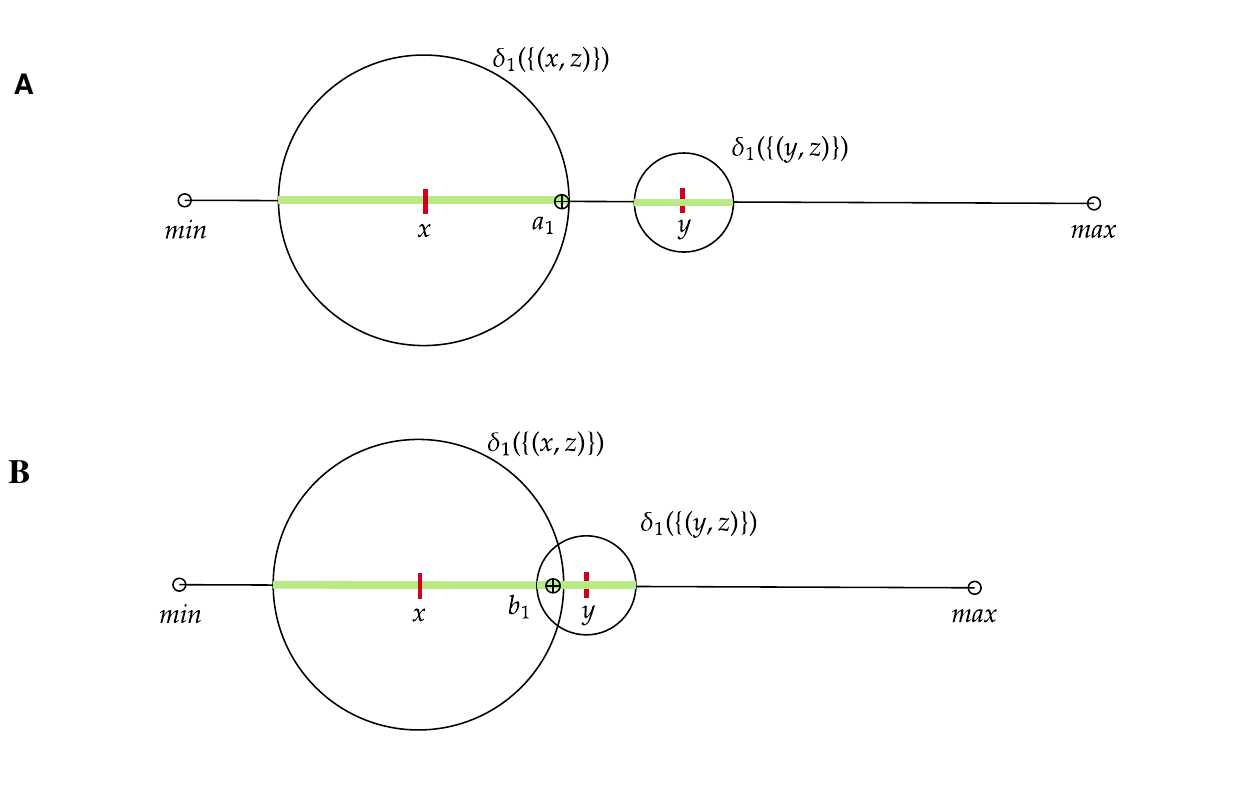}
    \caption{{\bf Case 4}  of the proof of Lemma \ref{key:Lemma}.\\
    We associate a ball with center $\alpha(u)$ of radius $\delta_1({u,v})$ for every $u \in \{x,y,min,max\}$ on $\mathcal{A}$ (resp. $\beta(u)$ on $\mathcal{B}$). 
Each ball associated to variable $u$ accounts for the ability of the pair $\{u,z\}$ to separate. The crossed dot on $\mathcal{A}$ represents the choices on the linear order from $F^k$, the crossed dot on $\mathcal{B}$ represent the best attempt at a response for indistinguishability. In {\bf Case 4}, a choice on $\mathcal{A}$ is in one ball and not in another, which cannot be matched by an element of $\mathcal{B}$.}
    \label{fig:case4}
\end{figure*}

\smallskip

Assume such an $a_i$ exists, we set $u,u' \in \{x,y,min,max\}$ and $b \in \mathcal{U}^{\mathcal{B}}$ such that d$(a_i, \alpha(u'))=$ d$(b, \beta(u'))$,  $<$-type$(a_i) =$ $<$-type$(b)$, note that such a $b$ exists by {\bf Assumption 2}.\\
By hypothesis d$(b, \beta(u)) <\text{d}(a_i, \alpha(u))$, so we must have that d$( \beta(u'), \beta(u)) \neq$ d$( \alpha(u'), \alpha(u))$. Notice that d$( \beta(u'), \beta(u)) \leq $ d$(b, \beta(u)) + $d$(b, \beta(u')) \leq \delta_1(\{z,u\}) + \delta_1(\{z,u'\}) \leq \delta(\{u,u'\}) $. So:

\begin{itemize}
    \item d$( \alpha(u'), \alpha(u)) \neq$ d$( \beta(u'), \beta(u))$;
    \item d$( \beta(u'), \beta(u)) \leq \delta(\{u,u'\}) $.
\end{itemize}

So $\delta (\{u,u'\})$ separates $\langle I,J \rangle$. \qed

\smallskip
{\bf Assumption 4} In the rest of the proof, we will suppose that if there is an $a_i$ and $u,u' \in \{x,y,min,max\}$ such that d$(a_i, \alpha(u)) > \delta_1(\{z,u\})$ and d$(a_i, \alpha(u')) \leq \delta_1(\{z,u'\})$ then there is $b \in \mathcal{U}^{\mathcal{B}}$ such that d$(a_i, \alpha(u'))=$ d$(b, \beta(u'))$ and d$(b, \beta(u)) > \delta_1(\{z,u\})$.

\smallskip

Before proceeding to the next case, we define quantities, called gaps, to characterize the elements of a linear order not in any nearness relationship to any of the assignments. 

\begin{definition}[Gap]
    \label{definition:gap}
For $(\mathcal{U}, \gamma) \in \{(\mathcal{A}, \alpha),(\mathcal{B}, \beta)\}$, we derive three \footnote{ The three cases are analogous, but a unified formal definition appears to be less readable.} particular quantities called \emph{gaps}:

\noindent
\begin{multline*}
Gap_{(\mathcal{U}, \gamma)}(\{min,x\}) = \{ a \in [min_U, \gamma(x)] |\\
\; \forall u \in \{x,y, min, max\}, 
\text{d}(a, \gamma(u)) > \delta_1 (\{u,z\}) \}
\end{multline*}

\noindent
\begin{multline*}
Gap_{(\mathcal{U}, \gamma)}(\{y,max\}) = \{ a \in [\gamma(y), max_U] | \\  \; \forall u \in \{x,y, min, max\}, 
\text{d}(a, \gamma(u)) > \delta_1 (\{u,z\}) \}
\end{multline*}

\noindent
\begin{multline*}
Gap_{(\mathcal{U}, \gamma)}(\{x,y\}) = \{ a \in [\gamma(x), \gamma(y)] | \\  \; \forall u \in \{x,y, min, max\}, 
\text{d}(a, \gamma(u)) > \delta_1 (\{u,z\}) \}
\end{multline*}
\end{definition}

\begin{figure*}[h]
    \centering
    \includegraphics[width=0.8\linewidth]{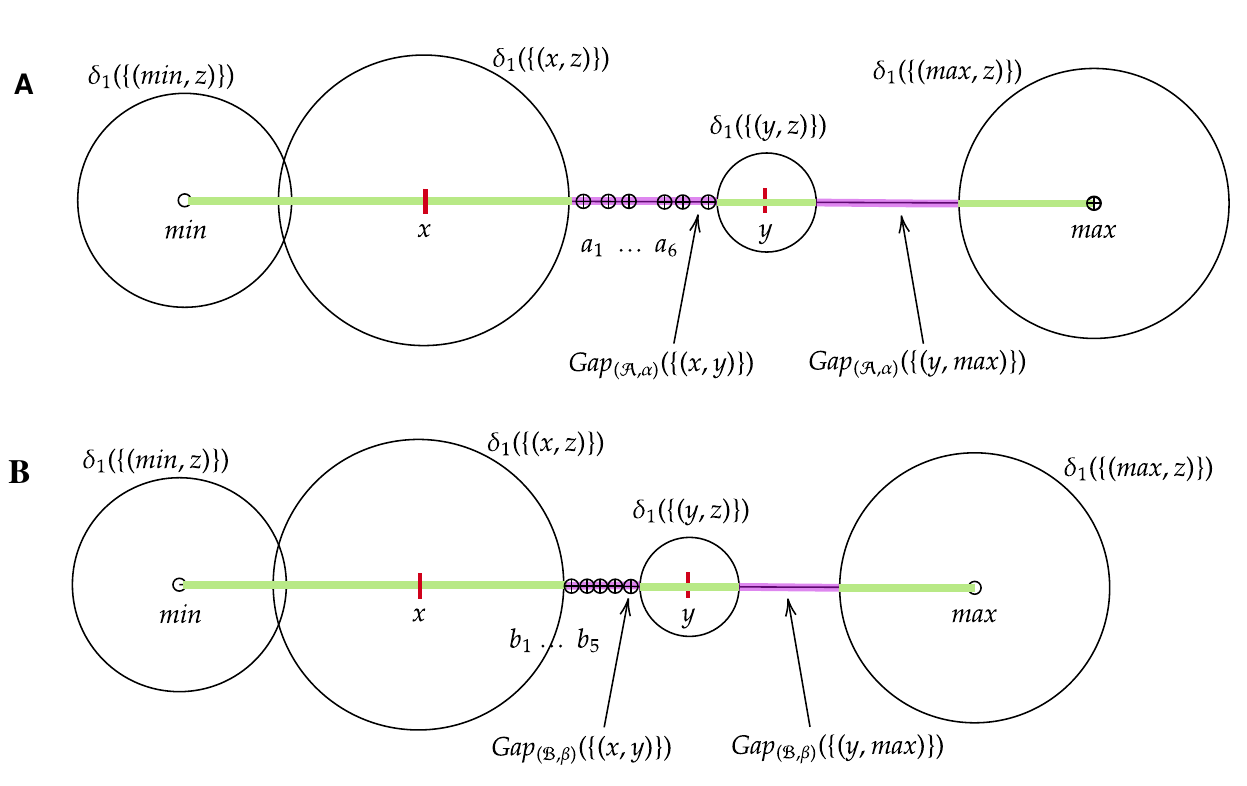}
    \caption{{\bf Case 5} of the proof of Lemma \ref{key:Lemma}.\\
    We associate a ball with center $\alpha(u)$ of radius $\delta_1({u,v})$ for every $u \in \{x,y,min,max\}$ on $\mathcal{A}$ (resp. $\beta(u)$ on $\mathcal{B}$). 
Each ball associated to variable $u$ accounts for the ability of the pair $\{u,z\}$ to separate. The regions of the linear orders that are not covered by any of the balls are the gaps, in purple on the Figure (in the example, the $\{min,x\}$ gap is empty. The crossed dots on $\mathcal{A}$ represents the choices on the linear order from $F^k$, the crossed dots on $\mathcal{B}$ represent the best attempts at a response for indistinguishability. In {\bf Case 5}, $p$ choices on $\mathcal{A}$ are in a gap, which cannot be matched by $p$ elements of $\mathcal{B}$ from the corresponding gap.}
    \label{fig:case5}
\end{figure*}

The motivation for gaps is that a separator cannot distinguish between two elements located within the same gap. We will refer to as ``gap variables" any two variables among $\{x,y,min,max\}$ that bound a given gap \footnote{A note on terminology: here $min$ and $max$ are seen as arguments of a separator function.}. More formally:

\begin{definition}[Gap variables]
\label{definition:gap_variables}
    For an interpretation $(\mathcal{U}, \gamma)$ and a pair of variables $(u,u') \in \{(min,x), (x,y), (y,max)\}$, the gap variables $(v,v')$ of $Gap_{(\mathcal{U}, \gamma)}(\{u,u'\})$ are defined as:
    \begin{itemize}
        \item $(u,u') = (min,x)$: $(v,v')$ is a pair of variables in\\
        $\{(min,x),(min,y),(min,max) \}$ such that
        \noindent
        \begin{multline*}
        |Gap_{(\mathcal{U}, \gamma)}(\{min,x\})| +  \delta_1(\{z, v\}) + \delta_1(\{z, x\}) \geq \\ \text{d}(\gamma(v'),\gamma(v))
        \end{multline*} 
        \item $(u,u') = (y,max)$: $(v,v')$ is a pair of variables in\\
        $\{(y, max),(x,max),(min,max) \}$ such that
        \begin{multline*}
            |Gap_{(\mathcal{U}, \gamma)}(\{y,max\})| +  \delta_1(\{z, v\}) + \delta_1(\{z, x\}) \geq \\ \text{d}(\gamma(v'),\gamma(v))
        \end{multline*}

        \item $(u,u') = (x,y)$: $(v,v')$ is a pair of variables in\\
        $\{(x,y),(min,y),(x,max),(min,max) \}$ such that
        \begin{multline*}
            |Gap_{(\mathcal{U}, \gamma)}(\{x,y\})| +  \delta_1(\{z, v\}) + \delta_1(\{z, x\}) \geq \\ \text{d}(\gamma(v'),\gamma(v))
        \end{multline*}

    \end{itemize}

\end{definition}
In the definition, the $(v,v')$ are ``bounding'' the gap $(u,u')$ on $U$. It is clear that such pairs $(v,v')$ exist by the definition of gaps, and if multiple pairs verify one of the inequalities, we pick the smallest pair in lexicographic order to be the gap variables.
We give a visual representation of the definitions in Fig. \ref{fig:gap_var} in Appendix \ref{appendix:gap_var}.

A gap $\{u,u'\}$ needs not have gap variables $\{u,u'\}$.  For example, the gap $\{min,x\}$ might have gap variables $\{min,y\}$ when the $\delta_1({y,v})$ radius circle contains the $\delta_1({x,v})$ radius circle, hence covering a bigger chunk of the $\{min,x\}$ interval.
If the gap is nonempty, the inequalities in Definition \ref{definition:gap_variables} become equalities, as the two circles corresponding to the gap variables are non-overlapping, and so the two radii added to the length of the gap add up to the length of the interval.

We derive a couple results for gap variables before proceeding to the next case:

\begin{proposition}
\label{proposition:var}
     If there are $u,u' \in \{x,y,min,max\}$ such that MIN $[$ d$(\alpha(u), \alpha(u')), $d$(\beta(u), \beta(u'))]  \leq \delta_1(\{z,u\})$ then d$(\alpha(u), \alpha(u')) = $d$(\beta(u), \beta(u'))$.
\end{proposition}

\begin{proof}
    Suppose not, since $\delta_1(\{z,u\}) \leq \delta(\{u,u'\})$ we have:
    \begin{itemize}
        \item d$(\alpha(u), \alpha(u')) \neq $d$(\beta(u), \beta(u'))$,
        \item MIN $[$ d$(\alpha(u), \alpha(u')), $d$(\beta(u), \beta(u'))]  \leq \delta(\{u,u'\})$.
    \end{itemize}
So $\delta (\{u,u'\})$ separates $\langle I,J \rangle$.

\end{proof}

\begin{proposition}
    Let $v,v'$ are the gap variables associated with $Gap_{(\mathcal{A}, \alpha)}(\{u,u'\})$ if and only if $v,v'$ are the gap variables associated with $Gap_{(\mathcal{B}, \beta)}(\{u,u'\})$.
\end{proposition}

\begin{proof}
    Let $v,v'$ be the gap variables associated with\\
 $Gap_{(\mathcal{A}, \alpha)}(\{u,u'\})$, and suppose $w,w'$ are the gap variables associated with  $Gap_{(\mathcal{B}, \beta)}(\{u,u'\})$.\\
 We will first reason from $(\mathcal{A}, \alpha)$ to $(\mathcal{B}, \beta)$, and suppose that $\alpha(v) < \alpha(w) \leq \alpha(u)$.
 By the definition of the gap on $(\mathcal{A}, \alpha)$, it must be that $\delta_1(\{z,v\}) \geq \text{d} (\alpha(w),\alpha(v)) + \delta_1(\{z,w\})$. Then proposition \ref{proposition:var} implies that d$(\alpha(w),\alpha(u)) = \text{d} (\beta(w),\beta(u))$. So $\delta_1(\{z,v\}) \geq \text{d} (\beta(w),\beta(v)) + \delta_1(\{z,w\})$, therefore $w$ cannot be a gap variable for $Gap_{(\mathcal{B}, \beta)}(\{u,u'\})$.\\
 The same reasoning applies from the gap on $(\mathcal{B}, \beta)$ to the gap on $(\mathcal{A}, \alpha)$, yielding the result.

\end{proof}

\smallskip

{\bf Case 5}: There is $(u,u') \in \{(min,x), (x,y), (y,max)\}$ such that there are $p$ choices of $F^k$ in $Gap_{(\mathcal{A}, \alpha)}$ and \\ $|Gap_{(\mathcal{B}, \beta)}(\{u,u'\})| < p $.

We give a visual representation of this case in Fig. \ref{fig:case5}.
\smallskip

By proposition \ref{proposition:var}, there is a unique pair $v,v'$ of gap variables for both $Gap_{(\mathcal{A}, \alpha)}(\{u,u'\})$ and $Gap_{(\mathcal{B}, \beta)}(\{u,u'\})$.
On one hand by the gap variables definition on $(\mathcal{B},\beta)$, we have d$(\beta(v'), \beta(v)) \leq$ $ \delta_1(\{z,v\}) +  \delta_1(\{z, v'\}) + |Gap_{(\mathcal{B}, \beta)}(\{u,u'\})| $ $\leq \delta_1(\{z,v\}) +  \delta_1(\{z, v'\}) + (k-1) \leq \delta (\{v,v'\})$.\\
On the other hand by the gap variables definition on $(\mathcal{A},\alpha)$, d$(\alpha(v'), \alpha(v)) =\delta_1(\{z,v\}) +  \delta_1(\{z, v'\}) + |Gap_{(\mathcal{A}, \alpha)}(\{u,u'\})| >$   d$(\beta(v'), \beta(v))$. Finally $\delta(\{v,v'\})$ separates $\langle I,J \rangle$.

\begin{figure*}[h]
  \centering
  \includegraphics[width=0.8\linewidth]{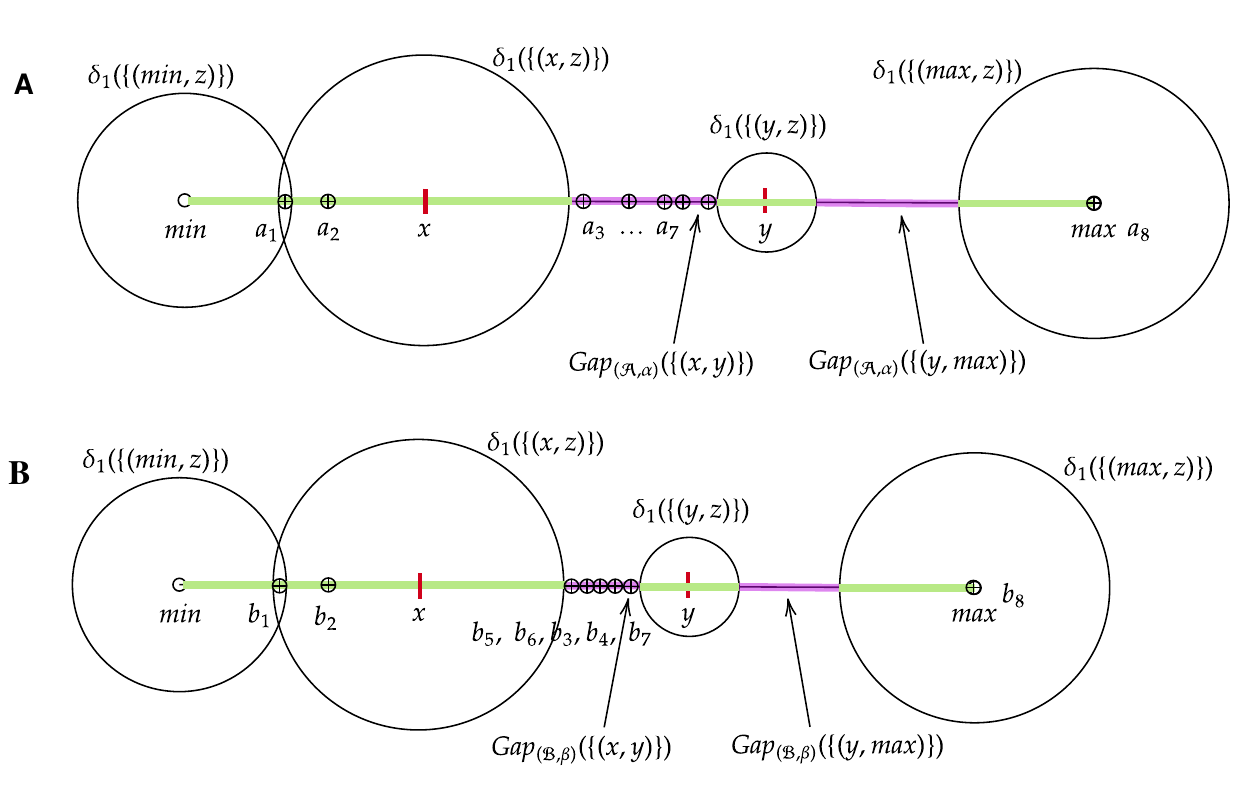}

\caption{ $<I_{F^k}, J_{G}>$ cannot be separated by $\delta_1$ in {\bf Case 6} of the proof of Lemma \ref{key:Lemma}.\\
We associate a ball with center $\alpha(u)$ of radius $\delta_1({u,v})$ for every $u \in \{x,y,min,max\}$ 
on $\mathcal{A}$ (resp. $\beta(u)$ on $\mathcal{B}$). 
Each ball associated to variable $u$ accounts for the ability of the pair $\{u,z\}$ to separate. The regions of the linear orders that are not covered by any of the balls are the gaps, in purple on the Figure (in the example, the $\{min,x\}$ gap is empty. The crossed dots on $\mathcal{A}$ represents the choices on the linear order from $F^k$, the crossed dots on $\mathcal{B}$ represent the best attempts at a response for indistinguishability. In {\bf Case 6}, every choice on $\mathcal{A}$ can be matched by a different $\delta_1$-indistinguishable element of $\mathcal{B}$.
}
\label{figure:proof_drawing}
\end{figure*}

\smallskip
{\bf Assumption 5} In the rest of the proof, we assume that if there are $p$ choices of $F^k$ in $Gap_{(\mathcal{A}, \alpha)}(\{u,u'\})$ for $(u,u') \in \{(min,x), (x,y), (y,max)\}$ then $|Gap_{(\mathcal{B}, \beta)}(\{u,u'\})| \geq p $.

\smallskip

{\bf Case 6}: Remaining cases.
Note that $S_i = \emptyset$  if and only if $a_i$ is in a gap. Recall the assumptions so far in the proof:

\begin{itemize}
    \item From {\bf Assumption 1}: for every $G^k$, $\langle I_{F^k},J_{G} \rangle$ can only be distinguished by $\delta_1(z,u)$ for $u \in \{x,y,min,max\}$;
    \item From {\bf Assumptions 2 to 4}: for every $a_i$ such that $S_i \neq \emptyset$, we have that there exists a unique $b_i \in \mathcal{U}^{\mathcal{B}}$ such that $<$-type$(a_i) = $ $<$-type$(b_i)$, d$(b_i, \beta(u)) = $ d$(a_i, \alpha(u))$ for all $u \in S_i$ and d$(b_i, \beta(u)) > \delta_1(\{z,u\})$ for all $u \in \{x,y,min, max\} \setminus S_i$.
    \item From {\bf Assumption 5}: for every $a_i$ in $Gap_{(\mathcal{A}, \alpha)}(\{u,u'\})$ for $(u,u') \in \{(min,x),(x,y),$ $(y,max)\}$, there exists a distinct $b_i \in Gap_{(\mathcal{B}, \beta)}(\{u,u'\})$.
\end{itemize}

Consider the choice function $G^k$ consisting of the $k$ $b_i$'s matching the $k$ $a_i$'s as described above. Suppose index $j$ from the choice function $G^k$ is selected to create $J_{G}$. Then by construction $\delta_1(\{z, u\})$ does not distinguish $\langle I_{F^k_j},J_{G} \rangle$ for any $u \in \{x,y,min,max\}$, which contradicts {\bf Assumption 1}. The situation is represented in Fig. \ref{figure:proof_drawing}.
In conclusion, the $\delta$ defined separates $\langle I,J \rangle$.
\end{proof}

%% file: paper/8_conclusion.tex
\section{Conclusion and perspectives} \label{sec:conc}

In this paper an EF
game is formulated for formula size for counting logic. It is used to prove a $\sqrt{n}/t$ lower bound for the size of 3-variable counting logic formulae with counting rank $t$, distinguishing a linear order of size $n$ from a larger one. 
The lower bound extends a $\Omega(\sqrt{n})$ lower bound of ~\cite{grohe_succinc}.
The proof is based on the approach of~\cite{grohe_succinc},
with a different argument for the central case of handling counting quantifiers.
Closing the gap between the lower bound and the upper bound of size $O(n/t)$ is an open problem. This is open even in the FO case where, as far as we know, no improvement is known of the linear upper bound.

The lower bound has some implications for comparing formula sizes of various fragments of counting logic.
Comparing the succinctness of various knowledge representation formalisms is studied in detail in knowledge compilation~\cite{Darwiche02}, Boolean complexity theory~\cite{Jukna7} and other areas, but, as noted in~\cite{grohe_succinc}, perhaps less so for predicate logic.  Bounds comparing $m$-variable counting logic formula sizes for $m = 2, 3, 4$ will be included in the final version of this paper. 
Separating the expressivity of the 2- and 3-variable fragments is an open problem, and it seems to be open for the FO case as well. 

We conclude with a brief description of the connection between \gnn and counting logic formula size, which is one of the motivations of this work and the subject of ongoing work. 
A \gnn works on a graph with feature vectors assigned to the nodes. In each round these are updated by applying a \textit{combination} function to the previous vector and an \textit{aggregate} of the feature vectors of the neighbors~\cite{Hamilton20}. A \textit{logical classifier} computes a unary query on graphs (e.g., assigning to every graph the set of red vertices with all blue neighbors).
Barcel\'o et al.~\cite{Barcelo20} showed that an FO logical classifier is computable by a \gnn iff it is definable in 
2-variable guarded counting logic. See also~\cite{logic_gnn,Grohe23}, also noting that this is a ``uniform'' model. Every such formula has a \gnn simulation with complexity (number of features and rounds) depending on the complexity of the formula.

In ongoing work with coauthors we study GNN learning a query corresponding to a formula. One question we consider is: to what extent can the model learned on one class of graphs be transferred to some other class of graphs? How does performance depend on the formula and the graph classes? 
The answer seems to be related to the connection between counting logic and the Weisfeiler-Leman algorithm mentioned in the introduction, and formula complexity may play a role here.
Another related question, recently considered in~\cite{Plusk}, is whether the underlying formula be extracted from a learned \gnn using explainability techniques?

%% file: paper/Appendix.tex
\appendices

\section{Proof of Lemma \ref{lemma:separator} parts 1 and 2}
\label{appendix:lemma_parts_1_2}
\subsection{Part 1}

Let $v$ be a leaf of $T$, by definition \ref{definition:syntax_tree}, $v$ if of the form $T_{\psi}^{\langle C,D \rangle}$ where $\psi$ is an atomic formula, say $R(u,u')$ where $u,u' \in \{x,y,min,max\}$ and $R \in \{<,=,succ\}$. Since $ C,D \models R(u,u')$, it is easy to verify that for all $(\mathcal{A}, \alpha) \in C $ and $(\mathcal{B}, \beta) \in D$:
\begin{enumerate}
    \item $<$-type$(\alpha(u), \alpha(u')) \neq $$<$-type$(\beta(u), \beta(u'))$  or
    \item both:
    \begin{itemize}
        \item MIN$[\text{d}(\alpha(u), \alpha(u' )) , \text{d}(\beta(u), \beta(u'))] \leq 1$ and,
        \item $\text{d}(\alpha(u), \alpha(u' )) \neq \text{d}(\beta(u), \beta(u' ))$.
\end{itemize}
\end{enumerate}
First if $u\neq u'$, consider the separator of $(C,D)$, $\delta_1$, defined as:
\begin{align*}
    \delta_1 = \begin{cases}
        1 \; \text{ if } p = \{u,u'\}\\
        0 \; \text{ otherwise.}
    \end{cases}
\end{align*}
Since $w(\delta_1)=1$, any minimal separator of $C<D$ has weight at most $1$.

Now, if $u=u'$, then $C= \emptyset$ or $D=\emptyset$, and then $\delta_1=0$ is a separator for $C,D$.

Finally, $w(\delta)\leq 1$ for $\delta$ a minimal separator of $v$.
\subsection{Part 2}

Let $v$ be a node of $T$ of the form $T_{\psi}^{\langle C,D \rangle}$ with two children $v_1$ and $v_2$. Let $\delta_1$ and $\delta_2$ be minimum separators of $v_1$ and $v_2$ respectively. We first prove the straightforward proposition:

\begin{proposition}
    $\delta' = \delta_1 + \delta_2$ is a separator of $v$.
\end{proposition}

\begin{proof}
   Let us fix $(\mathcal{A}, \alpha) \in C $ and $(\mathcal{B}, \beta) \in D$.  Suppose $\psi$ is of the form $\psi_1 \vee \psi_2$. The case $\psi$ of the form $\psi_1 \wedge \psi_2$ is symmetrical.

   By definition \ref{definition:syntax_tree}, there are $C_1, C_2$ such that $\delta_1$ is a separator of $T_{\psi_1}^{\langle C_1,D \rangle}$ and $\delta_2$ is a separator of $T_{\psi_2}^{\langle C_2,D \rangle}$ and with $C_1 \cup C_2 = C$.\\
   Whether $(\mathcal{A}, \alpha) \in C_1$ or $(\mathcal{A}, \alpha) \in C_2$, since $\delta' \geq \delta_1,\delta_2$ it is clear that $\delta'$ defined as above satisfies the separation condition:

   \begin{enumerate}
    \item <-type$(\alpha(u), \alpha(u')) \neq $<-type$(\beta(u), \beta(u'))$  or
    \item both:
    \begin{itemize}
        \item MIN$[\text{d}(\alpha(u), \alpha(u' )) , \text{d}(\beta(u), \beta(u'))] \leq \delta'(\{u,u'\})$ and,
        \item $\text{d}(\alpha(u), \alpha(u' )) \neq \text{d}(\beta(u), \beta(u' ))$.
\end{itemize}
\end{enumerate}
   Therefore $\delta'$ is a separator of $v$.
\end{proof}
We conclude with the upper bounding of the weight by showing that $w(\delta') \leq w(\delta_1)+w(\delta_2)$. By definition \ref{definition:weight}, it is easy to check that both $b(\delta') \leq b(\delta_1) + b(\delta_2)$ and $c(\delta') \leq c(\delta_1) + c(\delta_2)$ and therefore:

\begin{align*}
    w(\delta')^2 &= c(\delta')^2  + b(\delta') \\
                &\leq (c(\delta_1) + c(\delta_2)^2 + b(\delta_1) + b(\delta_2)\\
                & = c(\delta_1)^2 + c(\delta_2)^2 + b(\delta_1) + b(\delta_2) + 2c(\delta_1)c(\delta_2)\\
                &\leq w(\delta_1)^2 + w(\delta_2)^2 + 2 w(\delta_1)w(\delta_2)\\
                &= (w(\delta_1)+w(\delta_2))^2
\end{align*}
\qed

\section{End of the proof of Lemma \ref{lemma:separator} part 3.}
\label{appendix:consequence_weight}
\begin{figure*}[h]
    \centering
    \includegraphics[width=0.8\linewidth]{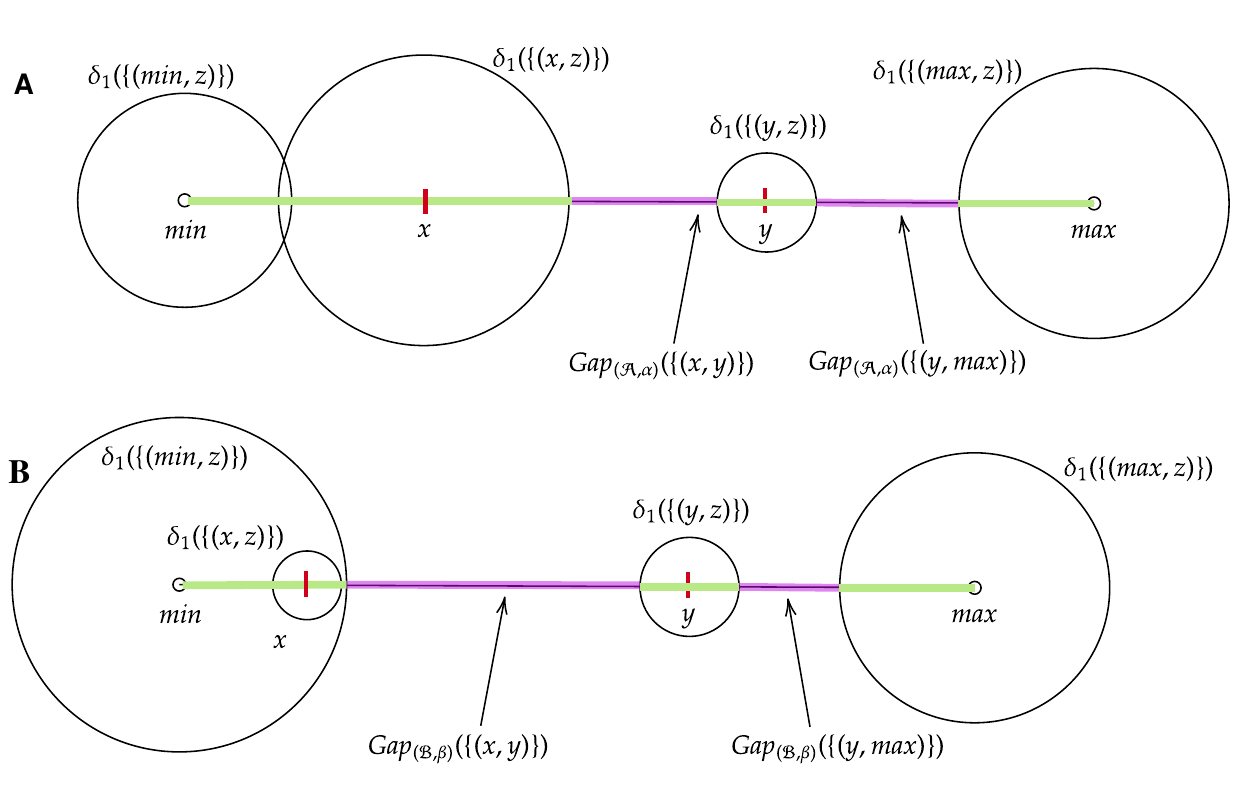}
    \caption{An illustration of the gap variables in a separating case.}
    \label{fig:gap_var}
\end{figure*}
We derive Lemma \ref{lemma:separator} part 3. from the consequence on the separator weight of the key Lemma \ref{key:Lemma}:
\begin{proof}
Suppose $v$ be a node of $T$ that has exactly one child $v_1$, we consider $\delta$ a minimal separator of $il(v)$ and $\delta_1$ be a minimal separator of $<A_1,B_1> := il(v_1)$.\\
We want to show that $w(\delta) \leq w(\delta_1) + t$.
From the definition of a syntax tree we know that either $sl(v) = \neg$ or $sl(v) = Q^{\geq k} u$, for some $Q \in \{\exists, \forall\}$, $k \leq t$,
and $u \in \{x, y, z\}$.\\
Case 1: $sl(v) = \neg $.\\
In this case $il(v) = <B,A>$ and $\delta_1$ is a separator for $il(v)$, therefore  $w(\delta) \leq w(\delta_1) \leq w(\delta_1) + t$.\\
Case 2: $sl(v) = Q^{\geq k} u$.\\
We call $\delta'$ the separator for $il(v)$ defined in Lemma \ref{key:Lemma}. Since $\delta$ is a minimal separator for $il(v)$, $ w(\delta) \leq w(\delta')$. We will show that $w(\delta') \leq w(\delta_1) + k$ to finish the proof.\\
Let us choose $u', u''$ so that $\{x, y, z\} = \{u, u', u''\}$. From the definition of $\delta'$ we get the following:
\begin{itemize}
    \item $c(\delta') = \delta'(\{u', u''\}) \leq c(\delta_1) + k - 1$ 
    \item $\delta'(\{min, max\}) \leq b(\delta_1) + k - 1$
\end{itemize}
\begin{proposition}
\label{proposition:1}
For any $f,g\in  \{x, y, z\}$,\\
$\delta'(\{min, f\}) +\delta'(\{g, max\}) \leq b(\delta_1) + 2c(\delta_1) + 2 (k-1)$.
\end{proposition}
\begin{proof}
\begin{itemize}
    \item Suppose $f = u$  or $g=u$, then  $\delta'(\{min, f\}) = 0$ or $\delta'(\{g, max\}) = 0$. And so $\delta'(\{min, f\}) +\delta'(\{g, max\}) \leq MAX [b(\delta_1), b(\delta_1) + c(\delta_1) + k-1]$ by definition $\ref{definition:weight}$ and Lemma \ref{key:Lemma}.
    \item Otherwise we get  $f,g\in  \{x, y, z\} \setminus \{u\}$  and:
    \begin{itemize}
        \item $\delta'(\{min, f\})= MAX ( \delta_1(\{min, f\})), \delta_1(\{min, u\})) + \delta_1(\{u, f\}) + k-1$
        \item $\delta'(\{max, g\})= MAX ( \delta_1(\{max, g\})), \delta_1(\{max, u\})) + \delta_1(\{u, g\}) + k-1$
    \end{itemize}
So:
\begin{align*}
\delta'(\{mi&n, f\}) + \delta'(\{max, g\}) \leq \\
    &MAX[\delta_1(\{min, f\}) + \delta_1(\{max, g\}), \\
    &\delta_1(\{min, f\}) + \delta_1(\{max, u\})) + \delta_1(\{u, g\}) + k-1,\\
    &\delta_1(\{min, u\}) + \delta_1(\{u, f\}) + k-1 + \delta_1(\{max, g\}),\\
    &\delta_1(\{min, u\}) + \delta_1(\{u, f\})+ \delta_1(\{max, u\}) + \\
    &\delta_1(\{u, g\}) + 2(k-1)]
\end{align*}
Therefore:
$$\delta'(\{min, f\}) + \delta'(\{g, max\}) \leq b(\delta_1) + 2c(\delta_1) + 2 (k-1)$$.
\end{itemize}
\end{proof}
From Proposition \ref{proposition:1} we can upper bound the border distance: \begin{multline*}
    b(\delta') = MAX[ \delta'(\{min, f\}) + \delta'(\{g, max\}), \delta'(\{min, max\})]\\
    \leq b(\delta_1) + 2c(\delta_1) + 2 (k-1)
\end{multline*}
Finally we can conclude:
\begin{align*}
w(\delta')^2 &= c(\delta')^2 + b(\delta') \\\
&\leq (c(\delta_1) + k-1 )^2 + b(\delta_1) + 2c(\delta_1) + 2(k-1)\\
&= c(\delta_1)^2 + b(\delta_1) + 2kc(\delta_1) + k^2-1\\
&\leq w(\delta_1)^2 + 2 kw(\delta_1) + k^2
= (w(\delta_1) + k)^2.    
\end{align*}
As a final consequence:
$$ w(\delta) \leq w(\delta') \leq w(\delta_1) + k \leq w(\delta_1) + t.$$
\end{proof}

\section{Figure for gap variables illustration Definition}
\label{appendix:gap_var}
We describe the illustration of Fig. \ref{fig:gap_var}. 
We associate a ball with center $\alpha(u)$ of radius $\delta_1({u,v})$ for every $u \in \{x,y,min,max\}$ on $\mathcal{A}$ (resp. $\beta(u)$ on $\mathcal{B}$). 
Each ball associated to variable $u$ accounts for the ability of the pair $\{u,z\}$ to separate. The regions of the linear orders that are not covered by any of the balls are the gaps, in purple on the Figure.\\
In the example on $\mathcal{A}$, the $\{min,x\}$ gap is empty, the gap variables for the $\{x,y\}$ gap are $(x,y)$ and the gap variables for the $\{y,max\}$ gap are $(y,max)$.\\
On $\mathcal{B}$, the $\{min,x\}$ gap is empty. The gap variables for the $\{x,y\}$ gap on $\mathcal{B}$ are $(min,y)$, as the $\delta_1(\{min,z\})$ ball covers more of the gap on the left than the   $\delta_1(\{min,z\})$ ball. Finally on $\mathcal{B}$ the gap variables for the $\{y,max\}$ gap are $(y,max)$.

%% file: main.bbl
\begin{thebibliography}{10}

\bibitem{n!_lwb_fmulasize}
{\sc Adler, M., and Immerman, N.}
\newblock An \emph{n!} lower bound on formula size.
\newblock {\em {ACM} Trans. Comput. Log. 4}, 3 (2003), 296--314.

\bibitem{Barcelo20}
{\sc Barcel{\'{o}}, P., Kostylev, E.~V., Monet, M., P{\'{e}}rez, J., Reutter, J.~L., and Silva, J.~P.}
\newblock The logical expressiveness of graph neural networks.
\newblock In {\em 8th International Conference on Learning Representations, {ICLR} 2020\/} (2020).

\bibitem{opt_graph_id}
{\sc Cai, J., F{\"{u}}rer, M., and Immerman, N.}
\newblock An optimal lower bound on the number of variables for graph identification.
\newblock {\em Comb. 12}, 4 (1992), 389--410.

\bibitem{carmosino2023finer}
{\sc Carmosino, M., Fagin, R., Immerman, N., Kolaitis, P.~G., Lenchner, J., and Sengupta, R.}
\newblock A finer analysis of multi-structural games and beyond.
\newblock {\em CoRR abs/2301.13329\/} (2023).

\bibitem{Darwiche02}
{\sc Darwiche, A., and Marquis, P.}
\newblock A knowledge compilation map.
\newblock {\em J. Artif. Intell. Res. 17\/} (2002), 229--264.

\bibitem{ebbing_flum}
{\sc Ebbinghaus, H., and Flum, J.}
\newblock {\em Finite model theory}.
\newblock Perspectives in Mathematical Logic. Springer, 1995.

\bibitem{Ehrenfeucht1961}
{\sc Ehrenfeucht, A.}
\newblock An application of games to the completeness problem for formalized theories.
\newblock {\em Fundamenta Mathematicae\/} (1961).

\bibitem{Etessami95}
{\sc Etessami, K.}
\newblock Counting quantifiers, successor relations, and logarithmic space.
\newblock In {\em Proceedings of the Tenth Annual Structure in Complexity Theory Conference, Minneapolis, Minnesota, USA, June 19-22, 1995\/} (1995), {IEEE} Computer Society, pp.~2--11.

\bibitem{Etessami97}
{\sc Etessami, K.}
\newblock Counting quantifiers, successor relations, and logarithmic space.
\newblock {\em J. Comput. Syst. Sci. 54}, 3 (1997), 400--411.

\bibitem{fagin2022number}
{\sc Fagin, R., Lenchner, J., Vyas, N., and Williams, R.~R.}
\newblock On the number of quantifiers as a complexity measure.
\newblock In {\em 47th International Symposium on Mathematical Foundations of Computer Science, {MFCS} 2022, August 22-26, 2022, Vienna, Austria\/} (2022), S.~Szeider, R.~Ganian, and A.~Silva, Eds., vol.~241 of {\em LIPIcs}, Schloss Dagstuhl - Leibniz-Zentrum f{\"{u}}r Informatik, pp.~48:1--48:14.

\bibitem{Fraisse54}
{\sc Fraisse, R.}
\newblock Sur quelques classifications des systemes de relations.

\bibitem{logic_gnn}
{\sc Grohe, M.}
\newblock The logic of graph neural networks.
\newblock In {\em 36th Annual {ACM/IEEE} Symposium on Logic in Computer Science, {LICS}\/} (2021), {IEEE}, pp.~1--17.

\bibitem{Grohe23}
{\sc Grohe, M.}
\newblock The descriptive complexity of graph neural networks.
\newblock In {\em {LICS}\/} (2023), pp.~1--14.

\bibitem{grohe_succinc}
{\sc Grohe, M., and Schweikardt, N.}
\newblock The succinctness of first-order logic on linear orders.
\newblock {\em Log. Methods Comput. Sci. 1}, 1 (2005).

\bibitem{Hajnal87}
{\sc Hajnal, A., Maass, W., Pudl{\'{a}}k, P., Szegedy, M., and Tur{\'{a}}n, G.}
\newblock Threshold circuits of bounded depth.
\newblock In {\em 28th Annual Symposium on Foundations of Computer Science\/} (1987), {IEEE} Computer Society, pp.~99--110.

\bibitem{Hajnal93}
{\sc Hajnal, A., Maass, W., Pudl{\'{a}}k, P., Szegedy, M., and Tur{\'{a}}n, G.}
\newblock Threshold circuits of bounded depth.
\newblock {\em J. Comput. Syst. Sci. 46}, 2 (1993), 129--154.

\bibitem{Hamilton20}
{\sc Hamilton, W.~L.}
\newblock {\em Graph Representation Learning}.
\newblock Synthesis Lectures on Artificial Intelligence and Machine Learning. Morgan {\&} Claypool Publishers, 2020.

\bibitem{HELLA1996}
{\sc Hella, L.}
\newblock Logical hierarchies in {PTIME}.
\newblock {\em Inf. Comput. 129}, 1 (1996), 1--19.

\bibitem{Hella2012}
{\sc Hella, L., and V{\"{a}}{\"{a}}n{\"{a}}nen, J.}
\newblock The size of a formula as a measure of complexity.
\newblock In {\em Logic Without Borders - Essays on Set Theory, Model Theory, Philosophical Logic and Philosophy of Mathematics}, {\AA}.~Hirvonen, J.~Kontinen, R.~Kossak, and A.~Villaveces, Eds., vol.~5 of {\em Ontos Mathematical Logic}. De Gruyter, 2015, pp.~193--214.

\bibitem{Immerman81}
{\sc Immerman, N.}
\newblock Number of quantifiers is better than number of tape cells.
\newblock {\em J. Comput. Syst. Sci. 22}, 3 (1981), 384--406.

\bibitem{Immerman_desc}
{\sc Immerman, N.}
\newblock {\em Descriptive complexity}.
\newblock Graduate texts in computer science. Springer, 1999.

\bibitem{immerman1990describing}
{\sc Immerman, N., and Lander, E.}
\newblock Describing graphs: A first-order approach to graph canonization.
\newblock 1990.

\bibitem{Jukna7}
{\sc Jukna, S.}
\newblock {\em Boolean Function Complexity - Advances and Frontiers}, vol.~27 of {\em Algorithms and combinatorics}.
\newblock Springer, 2012.

\bibitem{Karchmer90}
{\sc Karchmer, M., and Wigderson, A.}
\newblock Monotone circuits for connectivity require super-logarithmic depth.
\newblock {\em {SIAM} J. Discret. Math. 3}, 2 (1990), 255--265.

\bibitem{Kuske17}
{\sc Kuske, D., and Schweikardt, N.}
\newblock First-order logic with counting: At least, weak hanf normal forms always exist and can be computed!
\newblock {\em CoRR abs/1703.01122\/} (2017).

\bibitem{book_libkin}
{\sc Libkin, L.}
\newblock {\em Elements of Finite Model Theory}.
\newblock Texts in Theoretical Computer Science. An {EATCS} Series. Springer, 2004.

\bibitem{Morris19}
{\sc Morris, C., Ritzert, M., Fey, M., Hamilton, W.~L., Lenssen, J.~E., Rattan, G., and Grohe, M.}
\newblock Weisfeiler and leman go neural: Higher-order graph neural networks.
\newblock In {\em The Thirty-Third {AAAI} Conference on Artificial Intelligence, {AAAI} 2019\/} (2019), {AAAI} Press, pp.~4602--4609.

\bibitem{otto_book}
{\sc Otto, M.}
\newblock {\em Bounded Variable Logics and Counting: {A} Study in Finite Models}, vol.~9 of {\em Lecture Notes in Logic}.
\newblock Cambridge University Press, 2017.

\bibitem{otto_counting_bisimulation}
{\sc Otto, M.}
\newblock Graded modal logic and counting bisimulation.
\newblock {\em CoRR abs/1910.00039\/} (2019).

\bibitem{Plusk}
{\sc Pluska, A., Welke, P., G{\"{a}}rtner, T., and Malhotra, S.}
\newblock Logical distillation of graph neural networks.
\newblock In {\em Proceedings of the 21st International Conference on Principles of Knowledge Representation and Reasoning, {KR} 2024\/} (2024).

\bibitem{gnn_first}
{\sc Scarselli, F., Gori, M., Tsoi, A.~C., Hagenbuchner, M., and Monfardini, G.}
\newblock The graph neural network model.
\newblock {\em {IEEE} Trans. Neural Networks 20}, 1 (2009), 61--80.

\bibitem{vinallsmeeth2024quantifierdepthquantifiernumber}
{\sc Vinall-Smeeth, H.}
\newblock From quantifier depth to quantifier number: Separating structures with k variables, 2024.

\bibitem{Xu19}
{\sc Xu, K., Hu, W., Leskovec, J., and Jegelka, S.}
\newblock How powerful are graph neural networks?
\newblock In {\em 7th International Conference on Learning Representations, {ICLR} 2019\/} (2019).

\end{thebibliography}
